\documentclass[12pt,onecolumn,draftclsnofoot]{IEEEtran}

\usepackage[final]{graphicx}
\usepackage[reqno]{amsmath}
\usepackage{amssymb}
\usepackage{amsthm}
\usepackage{subfig}
\usepackage{epstopdf}
\usepackage{algorithm}
\usepackage{algorithmicx}
\usepackage{algpseudocode}
\usepackage{setspace}
\usepackage[T1]{fontenc}

\newtheorem{thm}{Theorem}

\newtheorem{lem}{Lemma}

\newtheorem{defn}{Definition}
\newtheorem{remk}{Remark}

\begin{document}

\sloppy

\title{Fundamental Limits of Dynamic Interference Management with Flexible Message Assignments}
\author{{\large{Tolunay Seyfi, {\em Student Member, IEEE}, Yasemin Karacora, {\em Student Member, IEEE} and Aly El Gamal, {\em Member, IEEE}}}
\thanks{Tolunay Seyfi and Aly El Gamal are with the ECE Department at Purdue University, West Lafayette, IN (e-mail: tseyfi,elgamala@purdue.edu). Yasemin Karacora was with the ECE Department at Purdue University, and is now with the ECE Department of Ruhr-University Bochum (email: yasemin.karacora@rub.de)}
\thanks{This paper was presented in part at the 2013 and 2017 versions of the Asilomar Conference on Signals, Systems and Computers, Pacific Grove, CA.~\cite{ElGamal-Veeravalli-Asilomar13,Karacora-Seyfi-ElGamal-Asilomar17}}.
\thanks{This research was supported by Huawei grant no. 301689 to Purdue University.}
}

\maketitle
\begin{abstract}
The problem of interference management is considered in the context of a linear interference network that is subject to long term channel fluctuations due to shadow fading. The fading model used is one where each link in the network is subject independently to erasure with probability $p$. It is assumed that each receiver in the network is interested in one unique message, which is made available at $M$ transmitters. For the case where $M=1$, the cell association problem is considered, and for $M>1$, the problem of setting up the backhaul links for Coordinated Multi-Point (CoMP) transmission is investigated. In both cases, optimal schemes from a Degrees of Freedom (DoF) viewpoint are analyzed for the setting of no erasures, and new schemes are proposed with better average DoF performance at higher probabilities of erasure. Additionally, for $M=1$, the average per user DoF is characterized for every value of $p$, and optimal message assignments are identified. For $M>1$, it is first established that there is no strategy for assigning messages to transmitters in networks that is optimal for all values of $p$. The optimal cooperative zero-forcing scheme for $M=2$ is then identified, and shown to be information-theoretically optimal when the size of the largest subnetwork that contains no erased links is at most five.
\end{abstract}

\section{Introduction}
Modern wireless networks are limited by interference from other links, and managing this interference is essential in meeting the exponential growth in demand for wireless data services.
Much of the recent effort on interference management has considered special settings where a fixed model is assumed for the channel connectivity (see e.g.,~\cite{Veeravalli-ElGamal-Cambridge}). However, in practice, wireless network topologies change frequently because of user mobility. Moreover, because of the envisioned development of heterogeneous networks, where client devices may be enabled to serve as infrastructural nodes, network topologies are expected to change even more frequently to exploit opportunities for improved performance.

Our goal in this paper is to start developing an information-theoretic framework for analyzing dynamic interference networks. The task of this envisioned framework is two-fold. First, as the network connectivity is expected to change, the choices made for transmitter (and possibly receiver) selection, fractional reuse, Coordinated Multi-Point (CoMP) transmission and reception, and interference alignment and zero-forcing have to take into account statistical knowledge of these changes. More specifically, these choices can be vastly different from those made for any specific realization of the network. Second, as devices have to learn the topology of the network, there arises a tradeoff between exploring the structure of the network and exploiting current knowledge for efficient transceiver design. Since this is a first attempt, we only focus on the first task and consider the problem of assigning messages to transmitters in order to achieve optimal average performance in a single-hop dynamic linear interference network.

In~\cite{Mceliece-Stark-IT84}, the authors analyzed the average capacity for a point-to-point channel model where slow changes in the channel result in varying severity of noise. In this work, we apply a similar concept to interference networks by assuming that slowly changing deep fading conditions result in link erasures. We consider the linear interference network introduced by Wyner~\cite{Wyner}, with the consideration of two fading effects. Long-term fluctuations that result in link erasures over a complete block of time slots, and short-term fluctuations that allow us to assume that any specific joint realization for the non-zero channel coefficients, will take place with zero probability. We study the problem of associating receivers with transmitters and setting up the backhaul links for Coordinated Multi-Point (CoMP) transmission, in order to achieve the optimal average Degrees of Freedom (DoF). This problem was studied in~\cite{ElGamal-Annapureddy-Veeravalli-arXiv12} for the case of no erasures. The results in~\cite{ElGamal-Annapureddy-Veeravalli-arXiv12} also reveal that optimal transceiver strategies depend intimately on the interference network topology. Therefore, our goal in this work is to leverage these results for static interference networks to dynamic interference networks. We extend the schemes in~\cite{ElGamal-Annapureddy-Veeravalli-arXiv12} to consider the occurrence of link erasures, and propose new schemes that lead to achieving better average DoF at high probabilities of erasure.

\subsection{Motivating Example}  
The following example motivates the study of new strategies for interference management in dynamic interference networks.

\begin{figure}[htb]
  \centering 
\subfloat[]{\label{fig:wynernetworkfig}\includegraphics[height=0.283\textwidth]{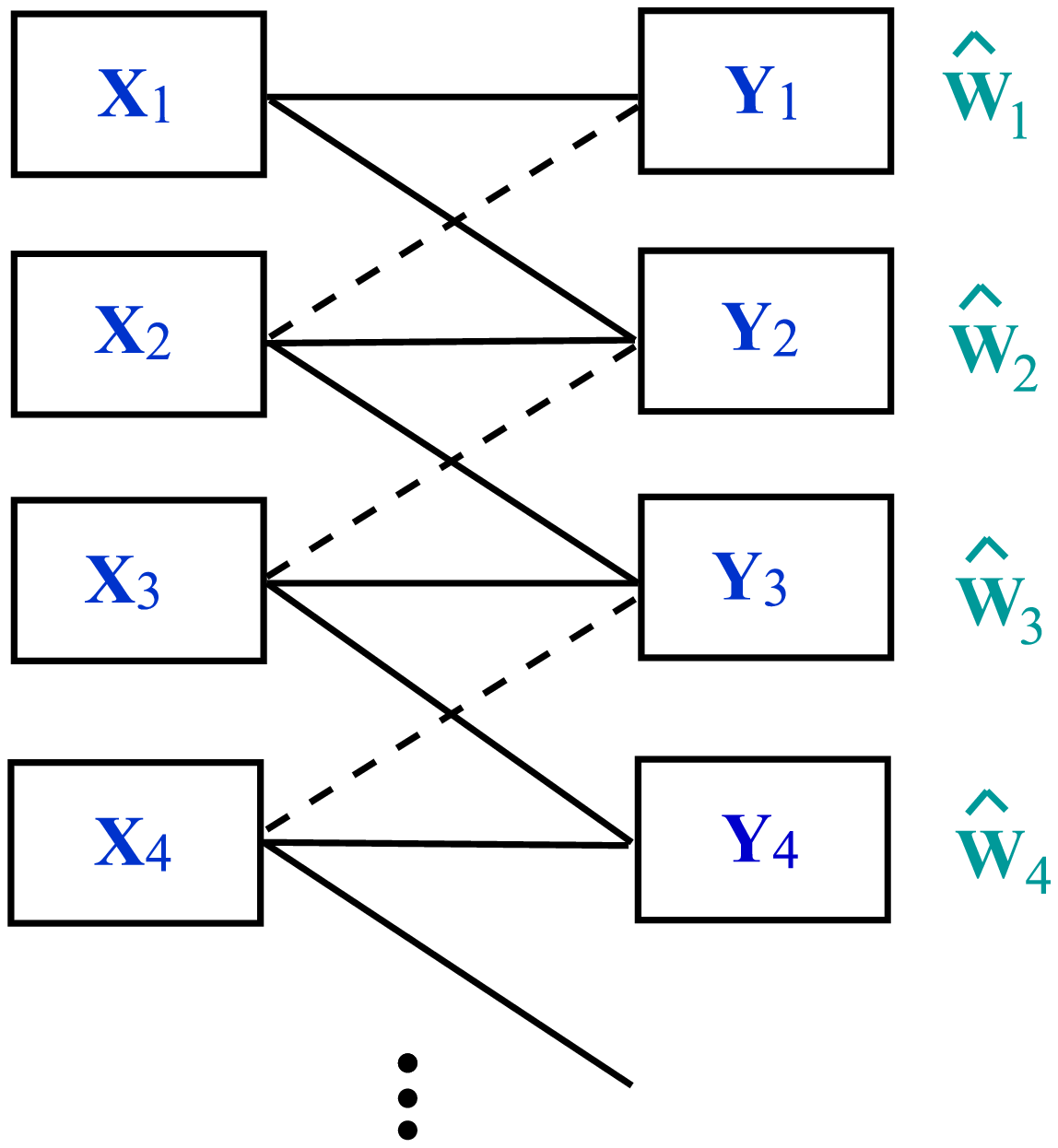}}                
\quad\quad\quad\quad\subfloat[]{\label{fig:wynernetworkaa}\includegraphics[width=0.293\textwidth]{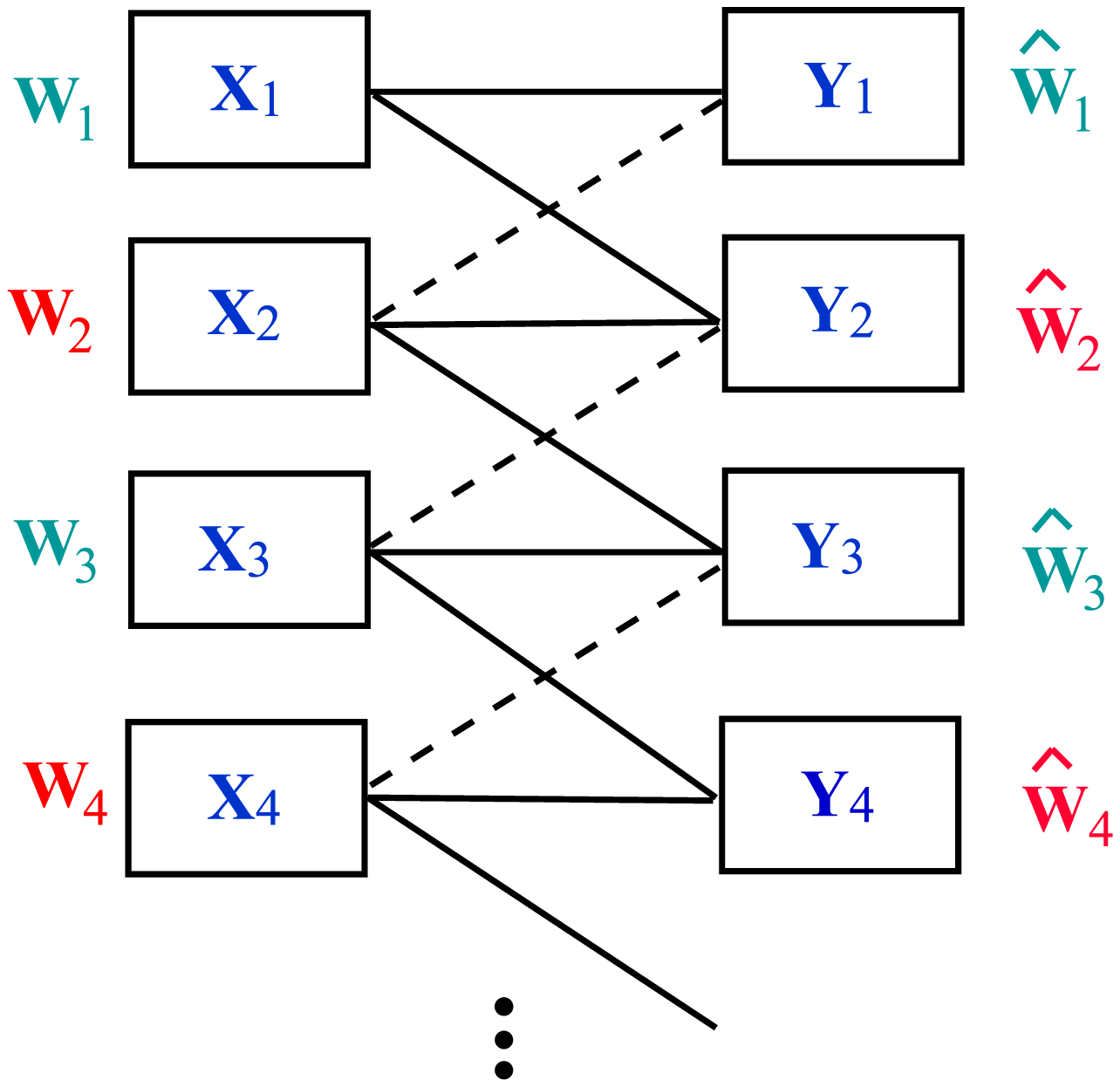}}
  \caption{\small A dynamic Wyner's interference network example. Solid edges always exist. All the dashed edges exist with probability $1-p$, or all do not exist with probability $p$. On the right hand side, messages indicated in red are not transmitted. A message assignment achieving $\frac{1}{2}$ per user DoF in both network realizations is shown. }
%
%
  \label{fig:wynernetwork}
\end{figure}

Consider the linear interference network depicted in Figure~\ref{fig:wynernetwork}, where each transmitter is connected to the receiver labelled with the same index as well as one following receiver. Also, with probability $1-p$, all transmitters will be connected to preceding receivers. If the dashed lines do not exist, then this channel is known as Wyner's asymmetric interference channel, and if they exist, it is known as Wyner's symmetric interference channel~\cite{Wyner}. The $i^{th}$ receiver is interested only in message (word) $W_i$. We further assume that each message can only be available at one transmitter, and we are required to find the message assignment that leads to maximizing the sum DoF averaged over all possible network realizations.  For the case of Wyner's symmetric interference channel, we know that the sum DoF of the network equals half the number of users (or a per user DoF of $\frac{1}{2}$). As shown in Figure~\ref{fig:wynernetworkaa}, the optimal per user DoF  can be achieved by assigning each message to the transmitter carrying the same index, and either activating all transmitter-receiver pairs with an odd index or those with an even index, with both choices being equally good with respect to the sum DoF criterion. For the case of Wyner's asymmetric interference channel, the same message assignment and transmission scheme can be used to achieve a per user DoF of $\frac{1}{2}$, and this remains the optimal per user DoF if the message assignment is fixed. The average per user DoF in this case across network topologies is thus $\frac{1}{2}$ regardless of the value of $p$.

In~\cite{ElGamal-Annapureddy-Veeravalli-ICC12}, and~\cite{Jafar-arXiv12}, the message assignment of Figure~\ref{fig:wynernetworkb} is shown to achieve $\frac{2}{3}$ per user DoF for the case of Wyner's asymmetric interference channel by activating transmitters $\{X_1,X_2,X_4,X_5,\ldots\}$ to transmit messages $\{W_1,W_3,W_4,W_6,\ldots\}$, with respect to order. It can be shown that the same message assignment can be used with the activation of all nodes in the network and using a simple modification of the asymptotic interference alignment scheme of~\cite{Cadambe-IA} to achieve $\frac{1}{2}$ per user DoF for Wyner's symmetric interference channel. The average per user DoF in this case is $\frac{3+p}{6} \geq \frac{1}{2}$, which can be shown to be the maximum achievable average per user DoF for this dynamic network. It follows that the message assignment in Figure~\ref{fig:wynernetworkb} is better than that of Figure~\ref{fig:wynernetwork}.

The above example shows that the consideration of network dynamics may lead to different design parameters and give rise to problems that were not considered before. For example, even though the message assignment of Figure~\ref{fig:wynernetworkb} is not uniquely optimal for Wyner's symmetric interference channel, the analysis of the optimal transmission scheme using this message assignment becomes important for the case where this particular assignment is optimal for other possible network realizations. 
Furthermore, this example is particularly relevant to the case when each message is allowed to be available at more than one transmitter. This models the downlink of a dynamic cellular network where the problem of assigning messages to transmitters corresponds to utilizing the backhaul in a way that maximizes the average performance of Coordinated Multi-Point transmission under network variations (see e.g.~\cite{CoMP-book} and~\cite{CoMP-industry} for a review of CoMP transmission techniques). In this work, we extend our previous study of CoMP transmission techniques to consider dynamic interference networks. We believe that the new framework can help find new insights that were not apparent in previously studied models.

\begin{figure}[!tb]
  \centering 
\subfloat[]{\label{fig:wynernetworkfigba}\includegraphics[height=0.24\textwidth]{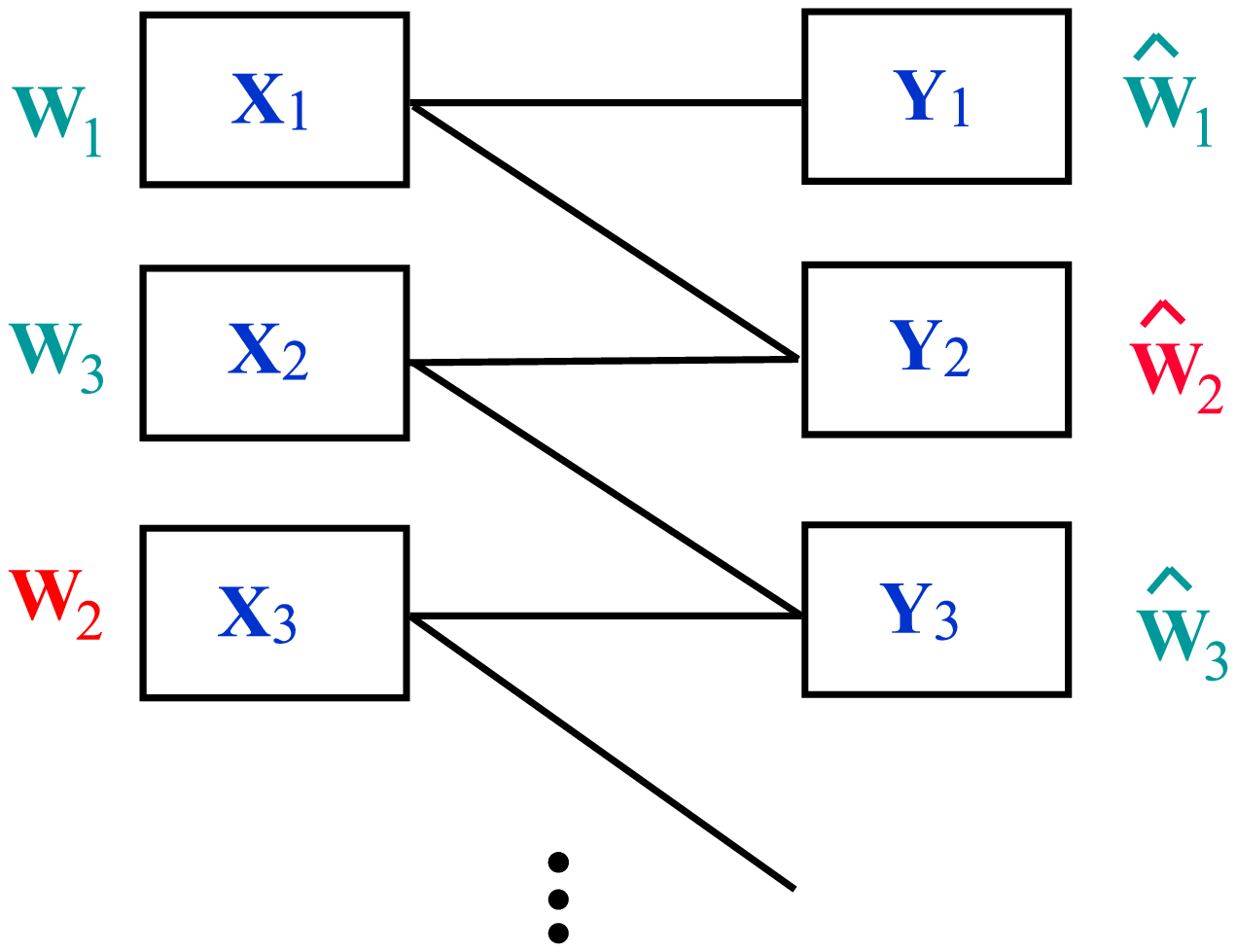}}                
\quad\quad\quad\quad\subfloat[]{\label{fig:wynernetworkbb}\includegraphics[width=0.32\textwidth]{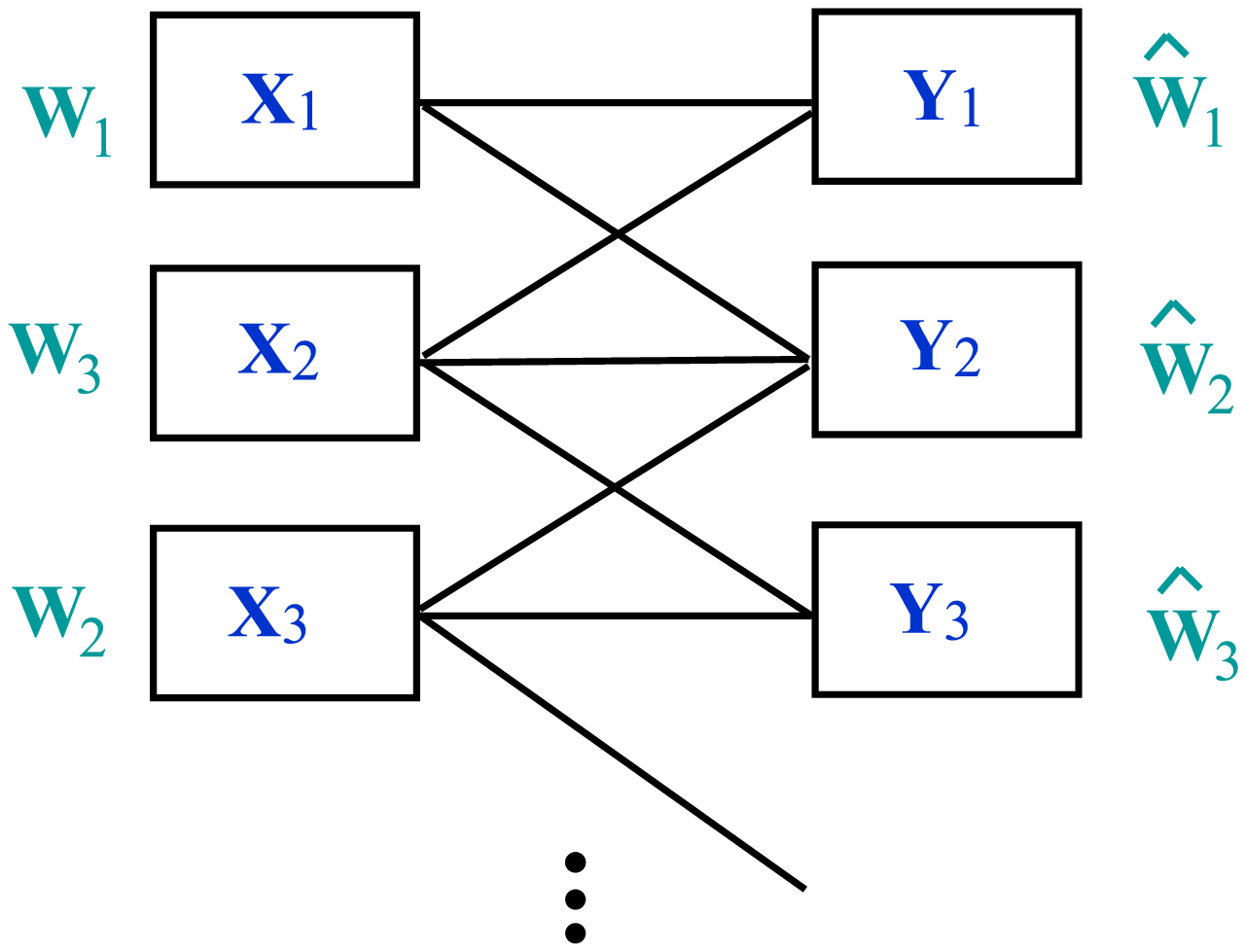}}
  \caption{\small Optimal message assignment for dynamic Wyner's interference network, achieving $\frac{2}{3}$ per user DoF when the dashed edges of Figure~\ref{fig:wynernetwork} do not exist, and $\frac{1}{2}$ per user DoF when they exist. On the right hand side, all transmitters are active and asymptotic interference alignment is used to achieve $\frac{1}{2}$ per user DoF. }
%
%
  \label{fig:wynernetworkb}
\end{figure}

\subsection{Document Organization}
We describe the system model and notation in Section~\ref{sec:systemmodel}. In Section~\ref{sec:cellassociation}, we consider the problem of assigning mobile transmitters to base station receivers (cell association) in a dynamic linear interference network. We then extend the analysis in Section~\ref{sec:comp} to a Coordinated Multi-Point transmission scenario where each message can be available at more than one transmitter. We finally present concluding remarks in Section~\ref{sec:conclusion}.

\subsection{Related Work}
At a fundamental level, the proposed work aims to significantly extend the notion of cognition in wireless networks. 
The term \emph{cognitive} has typically been used in literature to refer to a wireless transceiver that is aware of the messages/signals of another user in full or in part \cite{Devroye-Mitran-Tarokh-IT06,Lapidoth-Levy-Shamai-Wigger-arXiv12}. Our notion of cognition further allows wireless transceivers to be aware of the (changing) topology as well as the channel states of the network containing them.

In~\cite{Vahid-Aggarwal-Avestimehr-Sabharwal-arXiv12}, global knowledge of the network topology and local knowledge of the channel state information is assumed. It is worth noting here that although we assume global knowledge of the channel state information, the coding schemes described in this work make use of only local channel state information. Moreover, the knowledge of the topology is assumed to be available at the wireless transceivers but only statistics of the topology are taken into account when designing the infrastructural backhaul (assigning messages to transmitters). 

In practice, the channel coefficients are approximately estimated at the receivers by transmitting known pilot signals, and then they are fed back to the transmitters (see e.g.~\cite{PW09},~\cite{PW10} and~\cite{WMZ05a}). It is a common practice in information theoretic analysis to ignore the overhead of the estimation and communication of channel coefficients, in order to derive insights relevant to the remaining design parameters of the coding scheme; we follow this layered approach in this paper.

One major obstacle toward a practical implementation of asymptotic interference alignment is that the achievable DoF is approached only with a very large number of signal dimensions. The feasibility of alignment with finite symbol extension is studied in~\cite{finite-symbol-extension} and~\cite{Gomadam-Cadambe-Jafar-IT11}. All the presented coding schemes in this work do not require symbol extension.

In~\cite{ElGamal-Veeravalli-ISIT14}, it is shown that for linear interference networks with no erasures, the optimal schemes under the considered maximum number of transmitters per message constraint can be used to achieve the DoF under a more practical constraint that only limits the total backhaul load. The total backhaul load constraint is a constraint on the average transmit set size that allows for assigning some messages to a large number of transmitters at the cost of assigning others to fewer transmitters. In this work, we consider the maximum transmit set size constraint since it simplifies the combinatorial aspect of the problem and this is a first attempt at solving the problem in the dynamic interference network setting.

\section{System Model and Notation}\label{sec:systemmodel}
We use the standard model for the $K-$user interference channel with single-antenna transmitters and receivers,
\begin{equation}
Y_i(t) = \sum_{j=1}^{K} H_{i,j}(t) X_j(t) + Z_i(t),
\end{equation}
where $t$ is the time index, $X_j(t)$ is the transmitted signal of transmitter $j$, $Y_i(t)$ is the received signal at receiver $i$, $Z_i(t)$ is the zero mean unit variance Gaussian noise at receiver $i$, and $H_{i,j}(t)$ is the channel coefficient from transmitter $j$ to receiver $i$ over the time slot $t$. We remove the time index in the rest of the paper for brevity unless it is needed. 

We use $[K]$ to denote the set $\{1,2,\ldots,K\}$. Finally, for any set ${\cal A} \subseteq [K]$, we use the abbreviations $X_{\cal A}$, $Y_{\cal A}$, and $Z_{\cal A}$ to denote the sets $\left\{X_i, i\in {\cal A}\right\}$, $\left\{Y_i, i\in {\cal A}\right\}$, and $\left\{Z_i, i\in {\cal A}\right\}$, respectively. 
%
%

\subsection{Channel Model}
Each transmitter can only be connected to its corresponding receiver as well as one following receiver, and the last transmitter can only be connected to its corresponding receiver. More precisely,

\begin{equation}\label{eq:channel}
H_{i,j} \text{ is identically } 0 \text { iff } i \notin \{j,j+1\},\forall i,j \in [K].
\end{equation}

In order to consider the effect of long-term fluctuations (shadowing), we assume that communication takes place over blocks of time slots, and let $p$ be the probability of block erasure. In each block, we assume that for each $j$, and each $i \in \{j,j+1\}$, $H_{i,j}=0$ with probability $p$. Moreover, short-term channel fluctuations allow us to assume that in each time slot, all non-zero channel coefficients are drawn independently from a continuous distribution. Finally, we assume that global channel state information is available at all transmitters and receivers. 
\subsection{Message Assignment}
For each $i \in [K]$, let $W_i$ be the message intended for receiver $i$, and ${\cal T}_i \subseteq [K]$ be the transmit set of receiver $i$, i.e., those transmitters with the knowledge of $W_i$. The transmitters in ${\cal T}_i$ cooperatively transmit the message $W_i$ to the receiver $i$. The messages $\{W_i\}$ are assumed to be independent of each other. The \emph{cooperation order} $M$ is defined to be the maximum transmit set size:
\begin{equation}\label{eq:coop_order}
M = \max_i |{\cal T}_i|.
\end{equation}

\subsection{Message Assignment Strategy}\label{sec:strategy}

A message assignment strategy is defined by a sequence of transmit sets $({\cal T}_{i,K}), i\in[K], K\in\{1,2,\ldots\}$. For each positive integer $K$ and $\forall i\in[K]$,  ${\cal T}_{i,K} \subseteq [K], |{\cal T}_{i,K}| \leq M$. We use message assignment strategies to define the transmit sets for a sequence of $K-$user channels. The $k^{\mathrm{th}}$ channel in the sequence has $k$ users, and the transmit sets for this channel are defined as follows. The transmit set of receiver $i$ in the $k^{\mathrm{th}}$ channel in the sequence is the transmit set ${\cal T}_{i,k}$ of the message assignment strategy. 
\subsection{Degrees of Freedom}
The average power constraint at each transmitter is $P$. In each block of time slots, the rates $R_i(P)$ are achievable if the decoding error probabilities of all messages can be simultaneously made arbitrarily small as the block length goes to infinity, and this holds for almost all realizations of non-zero channel coefficients. The sum capacity $\mathcal{C}_{\Sigma}(P)$ is the maximum value of the sum of the achievable rates. The total number of degrees of freedom ($\eta$) is defined as $\limsup_{P \rightarrow \infty}\frac{ C_{\Sigma}(P)}{\log P}$. For a probability of block erasure $p$, we let $\eta_p$ be the average value of $\eta$ over possible choices of non-zero channel coefficients.

For a $K$-user channel, we define $\eta_p(K,M)$ as the best achievable $\eta_p$ over all choices of transmit sets satisfying the cooperation order constraint in \eqref{eq:coop_order}. In order to simplify our analysis, we define the asymptotic average per user DoF $\tau_p(M)$ to measure how $\eta_p(K,M)$ scale with $K$,
\begin{equation}
\tau_p(M) = \lim_{K\rightarrow \infty} \frac{\eta_p(K,M)}{K}.
\end{equation}

We call a message assignment strategy \emph{optimal} for a given erasure probability $p$,  if there exists a sequence of coding schemes achieving $\tau_p(M)$ using the transmit sets defined by the message assignment strategy. A message assignment strategy is $\emph{universally optimal}$ if it is optimal for all values of $p$.

\subsection{Interference Avoidance Schemes}
We pay special attention in this paper to the class of \emph{interference avoidance} schemes. For these schemes, each message is either not transmitted or allocated one degree of freedom.
Accordingly, every receiver is either active or inactive. An active receiver does not observe any interfering signals. For the case of no-cooperation i.e., $M=1$, we refer to these schemes as orthogonal TDMA schemes. The case where $M \geq 2$ corresponds to the scenario where cooperative zero-forcing can be used.

For a given erasure probability $p$, let $\tau_p^{(\text{TDMA})}$ and $\tau_p^{(\text{ZF})}(M\geq 2)$ be the average per user DoF under the restriction to orthogonal TDMA schemes, and cooperative zero-forcing transmit beamforming with cooperation order constraint $M$, respectively. 
\subsection{Subnetworks}
It will be useful in the rest of the paper to view each realization of the network where some links are erased, as a series of subnetworks that do not interfere. We formally define subnetworks below, but we first need to make the following definitions for a given message assignment and network realization.
\begin{defn}
For a given network realization and message assignment, we say that a {\bf message is enabled} if there exists a transmitter carrying the message and connected to its destined receiver.
\end{defn}
We also need to use the definition of \emph{irreducible message assignments} that is made in~\cite{ElGamal-Annapureddy-Veeravalli-arXiv12} for the case when no links can be erased (see also \cite[Chapter $6$]{Veeravalli-ElGamal-Cambridge}). The definition in~\cite{ElGamal-Annapureddy-Veeravalli-arXiv12} can be extended to the considered setting where links can be erased in the given realization, by replacing the condition $|x-y|\leq 1$, with the condition that the transmitters with indices $x$ and $y$ have to be connected to at least one common receiver. More precisely, for each user $i$, we construct a graph $G_i$ of $|{\cal T}_i|$ vertices with indices in ${\cal T}_i$, such that vertices $x,y \in {\cal T}_i$ are connected with an edge if transmitters $x$ and $y$ are connected to a common receiver. Further, vertices $i$ and $i-1$ are given a \emph{special mark} if $H_{i,i} \neq 0$ and $H_{i,i-1} \neq 0$, respectively. We then have the following definition.
\begin{defn}
We say that an assignment of message $W_i$ to transmitter $x$ is useful if the vertex $x$ is connected to a marked vertex in the graph $G_i$. 
\end{defn}
\begin{defn}
We say that a message assignment is irreducible, if for every user $i \in [K]$, the graph $G_i$ has only one component.
\end{defn}
The result in~\cite[Lemma $2$]{ElGamal-Annapureddy-Veeravalli-arXiv12} would then follow. In other words, a message assignment can be \emph{reduced} if we can remove one or more elements from the transmit set while guranteeing that the sum rate does not decrease due to this change. This is true if a transmitter carrying the message cannot be used either for delivering this message to its destination or for interference cancellation. We say that such assignments of messages to transmitters are not \emph{useful}. Since in our setting, the message assignment is based on the statistics of the network topology, it is possible that some message assignments can be reduced for a given realization of the topology.
\begin{defn}
For a given network realization and message assignment, we say that the {\bf message assignment is topology-reduced}, if we first remove all receivers whose messages are not enabled, and then for every transmit set, we remove all elements that are not useful, and the message assignment becomes irreducible for the considered realization.
\end{defn}
It follows from~\cite[Lemma $2$]{ElGamal-Annapureddy-Veeravalli-arXiv12} that topology-reduction of the message assignment cannot decrease the sum rate.
\begin{defn}
For a given network realization and message assignment, we say that a set of $k$ users with successive indices $\{i,i+1,\ldots,i+k-1\}$ form a {\bf subnetwork} if the following two conditions hold: 
\begin{enumerate}
\item The first condition is that $i=1$, or it is the case that for the topology-reduced version of the message assignment, either $W_i$ is not enabled or it is the case that there exists no message $W_x$, $x<i$ that is available at a transmitter connected to receiver $i$, and $W_i$ is not available at a transmitter connected to any receiver with an index $x < i$.

\item Secondly, $i+k-1=K$ or the first condition holds for $i+k$, i.e., $i+k$ is the first user in a new subnetwork.
\end{enumerate}
\end{defn}

\begin{defn}
We say that the {\bf subnetwork is atomic} if it does not contain smaller subnetworks.
\end{defn}
Note that for an atomic subnetwork, the transmitters carrying messages for users in the subnetwork have successive indices and for any transmitter $t$ carrying a message for a user in the subnetwork, and receiver $r$ in the subnetwork such that $r \in \{t,t+1\}$, the channel coefficient $H_{r,t} \neq 0$ (is not erased). 

\section{Cell Association}\label{sec:cellassociation}
We first consider the case where each receiver can be served by only one transmitter. This reflects the problem of associating mobile users with cells in a cellular downlink scenario. We start by discussing orthogonal schemes (TDMA-based) for this problem, and then show that the proposed schemes are optimal. 

For $i\in[K]$, let $N_i$ be the number of messages available at the $i^{th}$ transmitter, and let ${\bf N}^K=\left(N_1,N_2,\ldots,N_K\right)$. It is clear that the sequence ${\bf N}^K$ can be obtained from the transmit sets ${\cal T}_i, i\in[K]$; it is also true, as stated in the following lemma, that the converse holds. We borrow the notion of \emph{irreducible} message assignments from~\cite{ElGamal-Annapureddy-Veeravalli-arXiv12}. For $M=1$, an irreducible message assignment will have each message assigned to one of the two transmitters connected to its designated receiver.
\begin{lem}\label{lem:equiv}
For any irreducible message assignment where each message is assigned to exactly one transmitter, i.e., $|{\cal T}_i|=1, \forall i\in[K]$, the transmit sets ${\cal T}_i$, $i\in[K],$ are uniquely characterized by the sequence ${\bf N}^K$.
\end{lem}
\begin{proof}
Since each message can only be available at one transmitter, then this transmitter has to be connected to the designated receiver. More precisely, ${\cal T}_i \subset \{i-1,i\}, \forall i\in\{2,\ldots,K\}$, and ${\cal T}_1 = \{1\}$. It follows that each transmitter carries at most two messages and the first transmitter carries at least the message $W_1$, i.e., $N_i \in \{0,1,2\}, \forall i\in\{2,\ldots,K\}$, and $N_1 \in \{1,2\}$. Assume that $N_i=1, \forall i\in[K]$, then ${\cal T}_i=\{i\},\forall i\in[K]$. For the remaining case, we know that there exists $i \in \{2,\ldots,K\}$ such that $N_i=0$, since $\sum_{i=1}^{K} N_i = K$; we handle this case in the rest of the proof.

Let $x$ be the smallest index of a transmitter that carries no messages, i.e., $x = \min \{i: N_i=0\}$. We now show how to reconstruct the transmit sets ${\cal T}_i, i\in\{1,\ldots,x\}$ from the sequence $(N_1,N_2,\ldots,N_x)$. We note that ${\cal T}_i \in [x], \forall i\in[x]$, and since $N_x=0$, it follows that ${\cal T}_i \notin[x], \forall i\notin [x]$. It follows that $\sum_{i=1}^{x-1} N_i =x$. Since ${\cal T}_i \subset \{i-1,i\},\forall i\in\{2,\ldots,x\}$, we know that at most one transmitter in the first $x-1$ transmitters carries two messages. Since $\sum_{i=1}^{x-1} N_i =x$, and $N_i \in \{1,2\}, \forall i\in[x-1]$, it follows that there exists an index $y\in[x-1]$ such that $N_y = 2$, and $N_i=1, \forall i\in[x-1]\backslash\{y\}$. It is now clear that the $y^{th}$ transmitter carries messages $W_y$ and $W_{y+1}$, and each transmitter with an index $j\in\{y+1,\ldots,x-1\}$ is carrying message $W_{j+1}$, and each transmitter with an index $j \in \{1,\ldots,y\}$ is carrying message $W_j$. The transmit sets are then determined as follows. ${\cal T}_i=\{i\},\forall i\in[y]$ and ${\cal T}_i=\{i-1\},\forall i\in\{y+1,\ldots,x\}$.

We view the network as a series of subnetworks, where the last transmitter in each subnetwork is either inactive or the last transmitter in the network. If the last transmitter in a subnetwork is inactive, then the transmit sets in the subnetwork are determined in a similar fashion to the transmit sets ${\cal T}_i, i\in[x]$, in the above scenario. If the last transmitter in the subnetwork is the $K^{th}$ transmitter, and $N_K = 1$, then each message in this subnetwork is available at the transmitter with the same index.
\end{proof}

We use Lemma~\ref{lem:equiv} to describe message assignment strategies for large networks through repeating patterns of short ternary strings. Given a ternary string ${\bf S}=(S_1,\ldots,S_n)$ of fixed length $n$ such that $\sum_{i=1}^{n} S_i = n$, we define ${\bf N}^K$, $K \geq n$ as follows:
\begin{itemize}
\item $N_i=S_{i \text{ mod } n}$ if $\quad i\in\left\{1,\ldots,n\left\lfloor \frac{K}{n} \right\rfloor\right\}$,
\item $N_i=1$ if $i\in\left\{n\left\lfloor \frac{K}{n} \right\rfloor+1,\ldots,K\right\}$.
\end{itemize}

We now evaluate all possible message assignment strategies satisfying the cell association constraint using ternary strings through the above representation. We only restrict our attention to irreducible message assignments, and note that if there are two transmitters with indices $i$ and $j$ such that $i < j$ and each is carrying two messages, then there is a third transmitter with index $k$ such that $i < k < j$ that carries no messages. It follows that any string defining message assignment strategies that satisfy the cell association constraint, has to have one of the following forms:
\begin{itemize}
\item $S^{(1)}=(1)$,
\item $S^{(2)}=(2,1,1,\ldots,1,0)$,
\item $S^{(3)}=(1,1,\ldots,1,2,0)$,
\item $S^{(4)}=(1,1,\ldots,1,2,1,1,\ldots,1,0)$.
\end{itemize} 
Note that there is similarity between the message assignment strategies described by $S^{(2)}$ and $S^{(3)}$. In particular, any asymptotic per user DoF achievable through one of these strategies, is also achievable through the other. We hence introduce the three candidate message assignment strategies illustrated in Figure~\ref{fig:msgassignment}, and characterize the per user DoF achieved through each of them; we will formally show later that the optimal message assignment strategy at any value of $p$ is given by one of the three introduced strategies. We first consider the message assignment strategy defined by the string having the form $S^{(1)}=(1)$. Here, each message is available at the transmitter having the same index.

\begin{figure}
  \centering
\subfloat[]{\label{fig:highp}\includegraphics[height=0.2\textwidth]{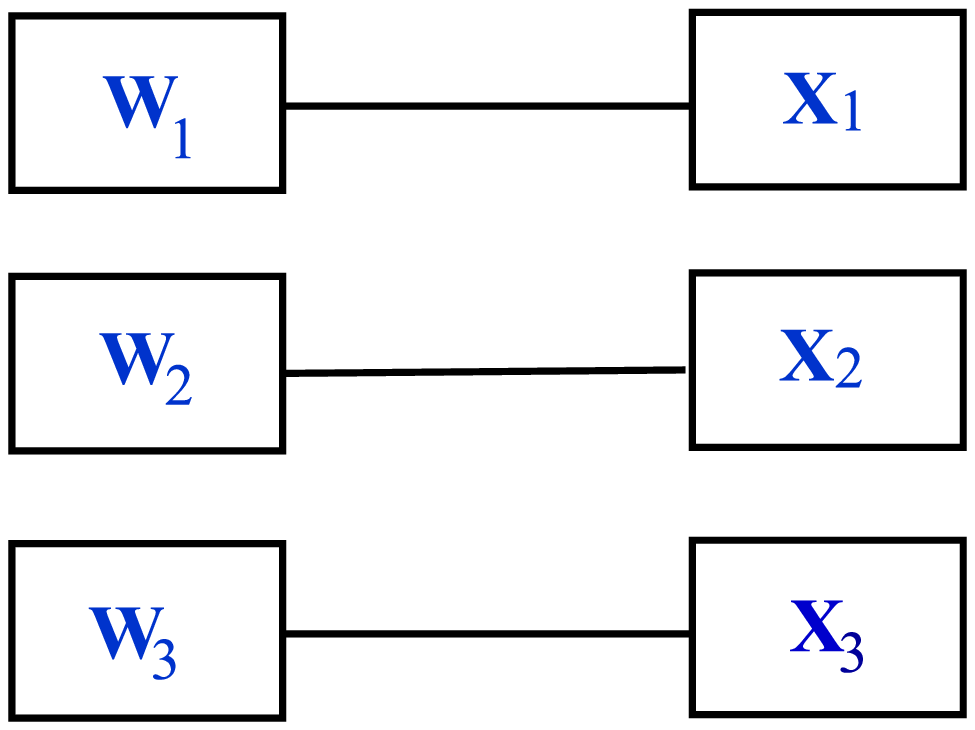}}                
\quad\quad\subfloat[]{\label{fig:lowp}\includegraphics[width=0.27\textwidth]{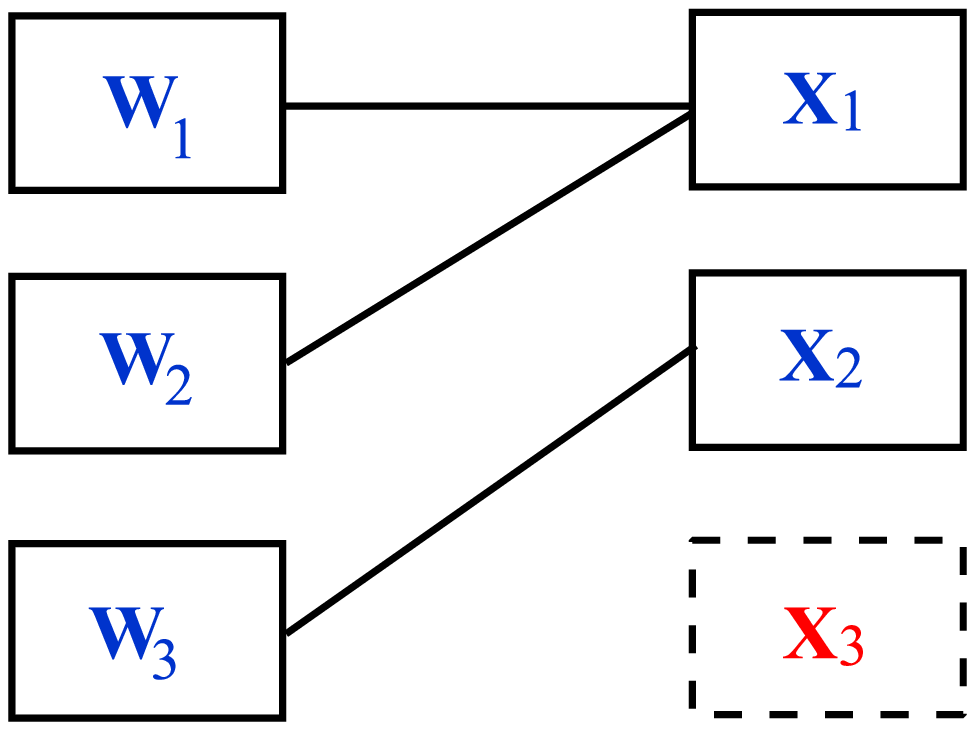}}
\quad\quad\subfloat[]{\label{fig:middlep}\includegraphics[width=0.27\textwidth]{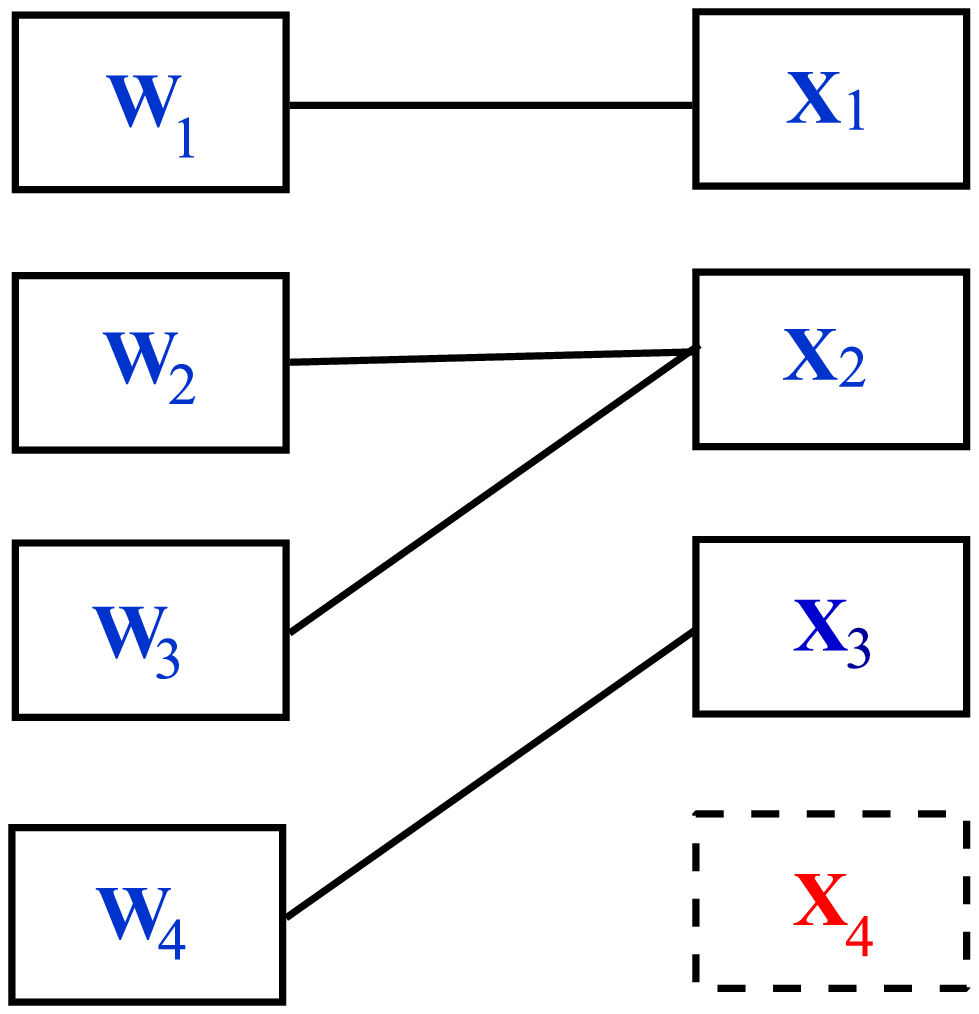}}
  \caption{The optimal message assignment strategies for the cell association problem. The red dashed boxes represent transmit signals that are inactive in all network realizations. The strategies in $(a)$, $(b)$, and $(c)$ are optimal at high, low, and middle values of the erasure probability $p$, respectively.}
  \label{fig:msgassignment}
\end{figure}
\begin{lem}\label{lem:highp}
Under the restriction to the message assignment strategy ${\cal T}_{i,K}=\{i\}, \forall K \in {\bf Z}^+, i\in[K],$ and orthogonal TDMA schemes, the average per user DoF is given by,
\begin{eqnarray}\label{eq:highp}
\tau_p^{(1)} &=& \frac{1}{2}\left(1-p+(1-p)\left(1-(1-p)^2\right)^2\right)\nonumber\\&&+\sum_{i=1}^{\infty} \frac{1}{2}\left(1-(1-p)^2\right)^2 (1-p)^{4i+1}.
\end{eqnarray}
\end{lem}
\begin{proof}
We will first explain a transmission scheme where $\frac{1}{2}\left(1-p+(1-p)\left(1-(1-p)^2\right)^2\right)$ average per user DoF is achieved, and then modify it to show how to achieve $\tau_p^{(1)}$. For each user with and odd index $i$, message $W_i$ is transmitted whenever the channel coefficient $H_{i,i} \neq 0$; the rate achieved by these users contributes to the average per user DoF by $\frac{1}{2}(1-p)$. For each user with an even index $i$, message $W_i$ is transmitted whenever the following holds: $H_{i,i}\neq0$, $W_{i-1}$ does not cause interference at $Y_i$, and the transmission of $W_i$ will not disrupt the communication of $W_{i+1}$ to its designated receiver; we note that this happens if and only if $H_{i,i} \neq 0 \text { and } \left(H_{i-1,i-1}=0 \text{ or } H_{i,i-1}=0\right)$$\text{ and}$
$(H_{i+1,i}=0 \text{ or } H_{i+1,i+1}=0)$. It follows that the rate achieved by users with even indices contributes to the average per user DoF by $\frac{1}{2} (1-p)\left(1-(1-p)^2\right)^2$.

 We now show a modification of the above scheme to achieve $\tau_p^{(1)}$. As above, users with odd indices have priority, i.e., their messages are delivered whenever their direct links exist, and users with even indices deliver their messages whenever their direct links exist and the channel connectivity allows for avoiding conflict with priority users. However, we make an exception to the priority setting in atomic subnetworks consisting of an odd number of users, and the first and last users have even indices; in these subnetworks, one extra DoF is achieved by allowing users with even indices to have priority and deliver their messages. The resulting extra term in the average per user DoF is calculated as follows. Fixing a user with an even index, the probability that this user is the first user in an atomic subnetwork consisting of an odd number of users in a large network is $\sum_{i=1}^{\infty}\left(1-(1-p)^2\right)^2 (1-p)^{4i+1}$; for each of these events, the sum DoF is increased by $1$, and hence the added term to the average per user DoF is equal to half this value, since every other user has an even index.

The optimality of the above scheme within the class of orthogonal TDMA-based schemes follows directly from~\cite[Theorem $1$]{Maleki-Jafar-arXiv13} for each realization of the network.
\end{proof}

We will show later that the above scheme is optimal at high erasure probabilities. In~\cite{ElGamal-Annapureddy-Veeravalli-arXiv12}, the optimal message assignment for the case of no erasures was characterized. The per user DoF was shown to be $\frac{2}{3}$, and was achieved by deactivating every third transmitter and achieving $1$ DoF for each transmitted message. We now consider the extension of this message assignment illustrated in Figure~\ref{fig:lowp}, which will be shown later to be optimal for low erasure probabilities. 
\begin{lem}\label{lem:lowp}
Under the restriction to the message assignment strategy defined by the string $S=(2,1,0)$, and orthogonal TDMA schemes, the average per user DoF is given by,
\begin{eqnarray}\label{eq:lowp}
\tau_p^{(2)} &=& \frac{2}{3}\left(1-p\right)+\frac{1}{3}p\left(1-p\right)\left(1-\left(1-p\right)^2 \right).
\end{eqnarray}
\end{lem}
\begin{proof}
For each user with an index $i$ such that $\left(i \text{ mod } 3 = 0\right)$ or $\left(i \text{ mod } 3=1\right)$, message $W_i$ is transmitted whenever the link between the transmitter carrying $W_i$ and the $i^{th}$ receiver is not erased; these users contribute to the average per user DoF by a factor of $\frac{2}{3}\left(1-p\right)$. For each user with an index $i$ such that $\left(i \text{ mod } 3=2\right)$, message $W_i$ is transmitted through $X_{i-1}$ whenever the following holds: $H_{i,i-1} \neq 0$, message $W_{i-1}$ is not transmitted because $H_{i-1,i-1}=0$, and the transmission of $W_i$ will not be disrupted by the communication of $W_{i+1}$ through $X_i$ because $\left(H_{i,i}=0\right) \text{ or } \left(H_{i+1,i}=0\right)$; these users contribute to the average per user DoF by a factor of $\frac{1}{3}p\left(1-p\right)\left(1-\left(1-p\right)^2\right)$. Using the considered message assignment strategy, the TDMA optimality of this scheme follows from~\cite[Theorem $1$]{Maleki-Jafar-arXiv13} for each network realization.
\end{proof}

We now consider the message assignment strategy illustrated in Figure~\ref{fig:middlep}. We will show later that this strategy is optimal for a middle regime of erasure probabilities. 
\begin{lem}\label{lem:middlep}
Under the restriction to the message assignment strategy defined by the string $S=(1,2,1,0)$, and orthogonal TDMA schemes, the average per user DoF is given by,
\begin{eqnarray}\label{eq:middlep}
\tau_p^{(3)} &=& \frac{1}{2}\left(1-p\right)\nonumber\\&&+\frac{1}{4}\left(1-p\right)\left(1-\left(1-p\right)^2 \right)\left(1+p+\left(1-p\right)^3\right).\nonumber\\
\end{eqnarray}
\end{lem}
\begin{proof}
As in the proof of Lemma~\ref{lem:highp}, we first explain a transmission scheme achieving part of the desired rate, and then modify it to show how the extra term can be achieved. Let each message with an odd index be delivered whenever the link between the transmitter carrying the message and the designated receiver is not erased; these users contribute to the average per user DoF by a factor of $\frac{1}{2} \left(1-p\right)$. For each user with an even index $i$, if $i \text{ mod } 4=2$, then $W_i$ is transmitted through $X_i$ whenever the following holds: $H_{i,i} \neq 0$, message $W_{i+1}$ is not transmitted through $X_i$ because $H_{i+1,i} =0$, and the transmission of $W_i$ will not be disrupted by the communication of $W_{i-1}$ through $X_{i-1}$ because either $H_{i,i-1}=0$ or $H_{i-1,i-1} =0$; these users contribute to the average per user DoF by a factor of $\frac{1}{4}p\left(1-p\right)\left(1-\left(1-p\right)^2\right)$. For each user with an even index $i$ such that $i$ is a multiple of $4$, $W_i$ is transmitted through $X_{i-1}$ whenever $H_{i,i-1}\neq 0$, and the transmission of $W_i$ will not disrupt the communication of $W_{i-1}$ through $X_{i-2}$ because either $H_{i-1,i-1}=0$ or $H_{i-1,i-2}=0$; these users contribute to the average per user DoF by a factor of $\frac{1}{4}\left(1-p\right)\left(1-\left(1-p\right)^2\right)$. 

We now modify the above scheme to show how $\tau_p^{(3)}$ can be achieved. Since the $i^{th}$ transmitter is inactive for every $i$ that is a multiple of $4$, users $\{i-3,i-2,i-1,i\}$ are separated from the rest of the network for every $i$ that is a multiple of $4$, i.e., these users form a subnetwork. We explain the modification for the first four users, and it will be clear how to apply a similar modification for every following set of four users. Consider the event where message $W_1$ does not cause interference at $Y_2$, because either $H_{1,1}=0$ or $H_{2,1}=0$, and it is the case that $H_{2,2}\neq 0$, $H_{3,2} \neq 0$, $H_{3,3} \neq 0$, and $H_{4,3} \neq 0$; this is the event that users $\{2,3,4\}$ form an atomic subnetwork, and it happens with probability $\left(1-\left(1-p\right)^2\right)\left(1-p\right)^4$. In this case, we let messages $W_2$ and $W_4$ have priority instead of message $W_3$, and hence the sum DoF for messages $\{W_1,W_2,W_3,W_4\}$ is increased by $1$. It follows that an extra term of $\frac{1}{4}\left(1-\left(1-p\right)^2\right)\left(1-p\right)^4$ is added to the average per user DoF. 

The TDMA optimality of the illustrated scheme follows from~\cite[Theorem $1$]{Maleki-Jafar-arXiv13} for each network realization.
\end{proof}

In Figure~\ref{fig:monenorm}, we plot the values of $\frac{\tau_p^{(1)}}{1-p}$, $\frac{\tau_p^{(2)}}{1-p}$, and $\frac{\tau_p^{(3)}}{1-p}$, and note that $\max \left\{\tau_p^{(1)},\tau_p^{(2)},\tau_p^{(3)}\right\}$ equals $\tau_p^{(1)}$ at high probabilities of erasure, and equals $\tau_p^{(2)}$ at low probabilities of erasure, and equals $\tau_p^{(3)}$ in a middle regime. 
\begin{figure}[htb]
\centering
\includegraphics[width=1\columnwidth]{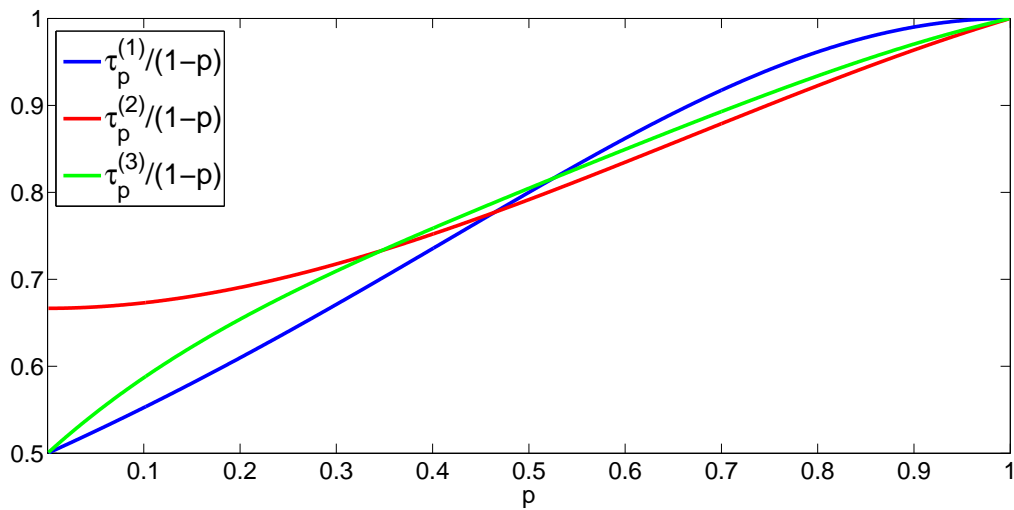}
\caption{The average per user DoF achieved through the strategies in Lemmas~\ref{lem:highp},~\ref{lem:lowp}, and~\ref{lem:middlep}, normalized by $(1-p)$.}
\label{fig:monenorm}
\end{figure} 

We now show that under the restriction to TDMA schemes, one of the message assignment strategies illustrated in Lemmas~\ref{lem:highp}, \ref{lem:lowp}, and \ref{lem:middlep} is optimal at any value of $p$. 
\begin{thm}\label{thm:tdma}
For any value $0 \leq p \leq 1$, the average per user DoF under restriction to orthogonal TDMA schemes is given as follows.
\begin{equation}\label{eq:tdma}
\tau_p^{(TDMA)}=\max \left\{\tau_p^{(1)},\tau_p^{(2)},\tau_p^{(3)}\right\},
\end{equation}
where $\tau_p^{(1)}$, $\tau_p^{(2)}$, and $\tau_p^{(3)}$ are given in~\eqref{eq:highp},~\eqref{eq:lowp}, and~\eqref{eq:middlep}, respectively.
\end{thm}
\begin{proof}
The inner bound follows from Lemmas~\ref{lem:highp},~\ref{lem:lowp}, and~\ref{lem:middlep}. In order to prove the converse, we need to consider all irreducible message assignment strategies where each message is assigned to a single transmitter. 
We know from Lemma~\ref{lem:highp} that the TDMA average per user DoF achieved through the strategy defined by the string of all ones having the form $S^{(1)}=(1)$ equals $\tau_p^{(1)}$, and hence the upper bound holds in this case. 

We now show that the TDMA average per user DoF achieved through strategies defined by strings of the form $S^{(2)}=\left(2,1,\ldots,1,0\right)$ is upper bounded by a convex combination of $\tau_p^{(1)}$ and $\tau_p^{(2)}$, and hence, is upper bounded by $\max \left\{\tau_p^{(1)},\tau_p^{(2)}\right\}$. The considered message assignment strategy splits each network into subnetworks consisting of a transmitter carrying two messages followed by a number of transmitters, each is carrying one message, and the last transmitter in the subnetwork carries no messages. We first consider the case where the number of transmitters carrying single messages is odd. We consider the simple scenario of the message assignment strategy defined by the string $(2,1,1,1,0)$, and then the proof will be clear for strategies defined by strings of the form $\left(2,1,1,\ldots,1,0\right)$ that have an arbitrary odd number of ones. In this case, it suffices to show that the average per user DoF in the first subnetwork is upper bounded by a convex combination of $\tau_p^{(1)}$ and $\tau_p^{(2)}$. The first subnetwork consists of the first five users; $W_1$ and $W_2$ can be transmitted through $X_1$. $W_3$, $W_4$ and $W_5$ can be transmitted through $X_2$, $X_3$, and $X_4$, respectively, and the transmit signal $X_5$ is inactive. 

We now explain the optimal TDMA scheme for the considered subnetwork. We first explain a simple scheme and then modify it to get the optimal scheme. Each of the messages $W_1$, $W_3$, and $W_5$ is delivered whenever the direct link between its carrying transmitter and its designated receiver is not erased. Message $W_2$ is delivered whenever message $W_1$ is not transmitted, and message $W_3$ is not causing interference at $Y_2$. Message $W_4$ is transmitted whenever $W_5$ is not causing interference at $Y_4$, and the transmission of $W_4$ through $X_3$ will not disrupt the communication of $W_3$. We now explain the modification; if there is an atomic subnetwork consisting of users $\{2,3,4\}$, then we switch the priority setting within this subnetwork, and messages $W_2$ and $W_4$ will be delivered instead of message $W_3$. The TDMA optimality of this scheme for each realization of the network follows from~\cite[Theorem $1$]{Maleki-Jafar-arXiv13}. Now, we note that the average sum DoF for messages $\{W_1,\ldots,W_5\}$ is equal to their sum DoF in the original scheme plus an extra term due to the modification. The average sum DoF for messages $\{W_1,W_2,W_5\}$ in the original scheme equals $3\tau_p^{(2)}$, and the sum of the average sum DoF for messages $\{W_3,W_4\}$ and the extra term is upper bounded by $2 \tau_p^{(1)}$. It follows that the average per user DoF is upper bounded by $\frac{2}{5} \tau_p^{(1)} + \frac{3}{5} \tau_p^{(2)}$. The proof can be generalized to show that the average TDMA per user DoF for message assignment strategies defined by strings of the form $S^{(2)}$ with an odd number of ones $n$, is upper bounded by $\frac{n-1}{n+2} \tau_p^{(1)} + \frac{3}{n+2} \tau_p^{(2)}$.

For message assignment strategies defined by a string of the form $S^{(2)}$ with an even number of ones $n$, it can be shown in a similar fashion as above that the TDMA average per user DoF is upper bounded by $\frac{n}{n+2} \tau_p^{(1)} + \frac{2}{n+2} \tau_p^{(2)}$. Also, for strategies defined by a string of the form $S^{(3)}=\left(1,1,\ldots,1,2,0\right)$ with a number of ones $n$, the TDMA average per user DoF is the same as that of a strategy defined by a string of the form $S^{(2)}$ with the same number of ones, and hence, is upper bounded by a convex combination of $\tau_p^{(1)}$ and $\tau_p^{(2)}$. Finally, for strategies defined by a string of the form $S^{(4)}=\left(1,1,\ldots,1,2,1,1,\ldots,1,0\right)$ with a number of ones $n$, it can be shown in a similar fashion as above that the average per user DoF is upper bounded by $\frac{n-2}{n+2}\tau_p^{(1)}+\frac{4}{n+2}\tau_p^{(3)}$.
\end{proof}

We now characterize the average per user DoF for the cell association problem by proving that TDMA schemes are optimal for any candidate message assignment strategy. In order to prove an information theoretic upper bound on the per user DoF for each network realization, we use Lemma $4$ from~\cite{ElGamal-Annapureddy-Veeravalli-arXiv12}, which we restate below. For any set of receiver indices ${\cal A} \subseteq [K]$, define $U_{\cal A}$ as the set of indices of transmitters that exclusively carry the messages for the receivers in ${\cal A}$, and the complement set is $\bar{U}_{\cal A}$. More precisely, $U_{\cal A} = [K]\backslash\cup_{i \notin {\cal A}} {\cal T}_i$.
\begin{lem}\cite[Lemma $4$]{ElGamal-Annapureddy-Veeravalli-arXiv12} \label{lem:dofouterbound}
If there exists a set ${\cal A}\subseteq [K]$, a function $f_1$, and a function $f_2$ whose definition does not depend on the transmit power constraint $P$, and $f_1\left(Y_{\cal A},X_{U_{\cal A}}\right)=X_{\bar{U}_{\cal A}}+f_2(Z_{\cal A})$, then the sum DoF $\eta \leq |{\cal A}|$.
\end{lem}
\begin{thm}\label{thm:mone}
The average per user DoF for the cell association problem is given by,

\begin{equation}\label{eq:tauone}
\tau_p\left(M=1\right)=\tau_p^{(TDMA)}=\max \left\{\tau_p^{(1)},\tau_p^{(2)},\tau_p^{(3)}\right\},
\end{equation}
where $\tau_p^{(1)}$, $\tau_p^{(2)}$, and $\tau_p^{(3)}$ are given in~\eqref{eq:highp},~\eqref{eq:lowp}, and~\eqref{eq:middlep}, respectively.
\end{thm}
\begin{proof}
In order to prove the statement, we need to show that $\tau_p(M=1) \leq \tau_p^{(TDMA)}$; we do so by using Lemma~\ref{lem:dofouterbound} to show that for any irreducible message assignment strategy satisfying the cell association constraint, and any network realization, the asymptotic per user DoF is given by that achieved through the optimal TDMA scheme.

Consider message assignment strategies defined by strings having one of the forms $S^{(1)}=(1)$, $S^{(2)}=\left(2,1,1,\ldots,1,0\right)$, and $S^{(3)}=\left(1,1,\ldots,1,2,0\right)$. We view each network realization as a series of atomic subnetworks, and show that for each atomic subnetwork, the sum DoF is achieved by the optimal TDMA scheme. For an atomic subnetwork consisting of a number of users $n$, we note that $\left\lfloor\frac{n+1}{2}\right\rfloor$ users are active in the optimal TDMA scheme; we now show in this case using Lemma~\ref{lem:dofouterbound} that the sum DoF for users in the subnetwork is bounded by $\left\lfloor\frac{n+1}{2}\right\rfloor$. Let the users in the atomic subnetwork have the indices $\{i,i+1,\ldots,i+n-1\}$, then we use Lemma~\ref{lem:dofouterbound} with the set ${\cal A}=\left\{i+2j: j\in\left\{0,1,2,\ldots,\left\lfloor\frac{n-1}{2}\right\rfloor\right\}\right\}$, except the cases of message assignment strategies defined by strings having one of the forms $S^{(1)}=(1)$ and $S^{(3)}=\left(1,1,\ldots,1,2,0\right)$ with an even number of ones, where we use the set ${\cal A}=\left\{i+1+2j: j\in\left\{0,1,2,\ldots,\frac{n-2}{2}\right\}\right\}$. We now note that each transmitter that carries a message for a user in the atomic subnetwork and has an index in $\bar{U}_{\cal A}$, is connected to a receiver in ${\cal A}$, and this receiver is connected to one more transmitter with an index in $U_{\cal A}$, and hence, the missing transmit signals $X_{\bar{U}_{\cal A}}$ can be recovered from $Y_{\cal A}-Z_{\cal A}$ and $X_{U_{\cal A}}$. The condition in the statement of Lemma~\ref{lem:dofouterbound} is then satisfied; allowing us to prove that the sum DoF for users in the atomic subnetwork is upper bounded by $|{\cal A}|=\left\lfloor\frac{n+1}{2}\right\rfloor$.

The proof is similar for message assignment strategies defined by strings that have the form $S^{(4)}=\left\{1,1,\ldots,1,2,1,1,\ldots,1,0\right\}$. However, there is a difference in selecting the set ${\cal A}$ for atomic subnetworks consisting of users with indices $\{i,i+1,\ldots,i+x,i+x+1,\ldots,i+n-1\}$, where $1 \leq x \leq n-2$, and messages $W_{i+x}$ and $W_{i+x+1}$ are both available at transmitter $i+x$. In this case, we apply Lemma~\ref{lem:dofouterbound} with the set ${\cal A}$ defined as above, but including indices $\{i+x,i+x+1\}$ and excluding indices $\{i+x-1,i+x+2\}$. It can be seen that the condition in Lemma~\ref{lem:dofouterbound} will be satisfied in this case, and the proved upper bound on the sum DoF for each atomic subnetwork, is achievable through TDMA.
\end{proof}
\begin{figure}[htb]
\centering
\includegraphics[width=1\columnwidth]{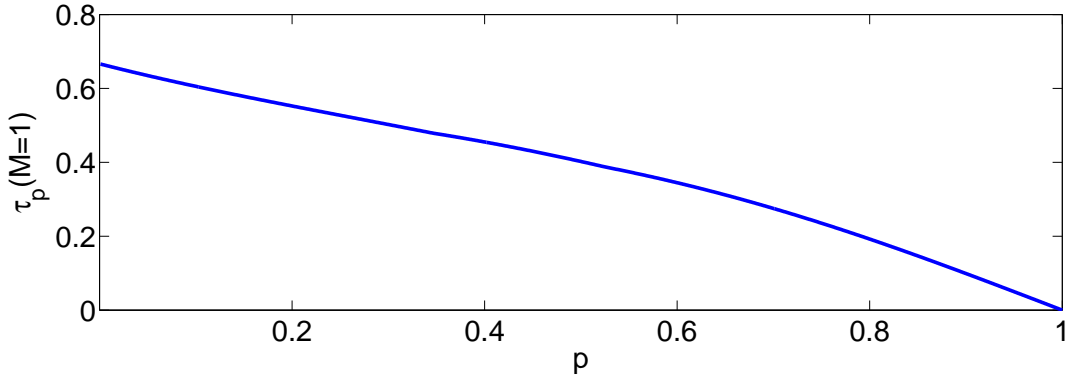}
\caption{The average per user DoF for the cell association problem}
\label{fig:monemax}
\end{figure} 

In Figure~\ref{fig:monemax}, we plot $\tau_p(M=1)$ at each value of $p$. The result of Theorem~\ref{thm:mone} implies that the message assignment strategies considered in Lemmas~\ref{lem:highp},~\ref{lem:lowp},~\ref{lem:middlep} are optimal at high, low, and middle values of the erasure probability $p$, respectively. We note that in densely connected networks at a low probability of erasure, the \emph{interference-aware} message assignment strategy in Figure~\ref{fig:lowp} is optimal; through this assignment, the maximum number of interference free communication links can be created for the case of no erasures. On the other hand, the linear nature of the channel connectivity does not affect the choice of optimal message assignment at high probability of erasure. As the effect of interference diminishes at high probability of erasure, assigning each message to a unique transmitter, as in the strategy in Figure~\ref{fig:highp}, becomes the only criterion of optimality. At middle values of $p$, the message assignment strategy in Figure~\ref{fig:middlep} is optimal; in this assignment, the network is split into four user subnetworks. In the first subnetwork, the assignment is optimal as the maximum number of interference free communication links can be created for the two events where there is an atomic subnetwork consisting of users $\{1,2,3\}$ or users $\{2,3,4\}$.

\section{Coordinated Multi-Point Transmission}\label{sec:comp}
We have shown that there is no message assignment strategy for the cell association problem that is optimal for all values of $p$. We first show in this section that this statement is true even for the case where each message can be available at more than one transmitter ($M>1$). Recall that for a given value of $M$, we say that a message assignment strategy is universally optimal if it can be used to achieve $\tau_p(M)$ for all values of $p$.
\begin{thm}\label{thm:comp}
For any value of the cooperation constraint $M \in {\bf Z}^+$, there does not exist a universally optimal message assignment strategy.
\end{thm}
\begin{proof}
The proof follows from Theorem~\ref{thm:mone} for the case where $M=1$. We show that for any value of $M>1$,  any message assignment strategy that enables the achievability of $\tau_p(M)$ at high probabilities of erasure, is not optimal for the case of no erasures, i.e., cannot be used to achieve $\tau_p(M)$ for $p=0$. For any message assignment strategy, consider the value of $\lim_{p \rightarrow 1} \frac{\tau_p(M)}{1-p}$ and note that this value equals the average number of transmitters in a transmit set that can be connected to the designated receiver. More precisely,
\begin{equation}\label{eq:cond}
\lim_{p \rightarrow 1} \frac{\tau_p(M)}{1-p}=\frac{\sum_{i=1}^{K}|{\cal T}_i \cap \{i-1,i\}|}{K},
\end{equation}
where ${\cal T}_i$ in~\eqref{eq:cond} corresponds to an optimal message assignment strategy at high probabilities of erasure. It follows that there exists a value $0 < \bar{p} < 1$ such that for any message assignment strategy that enables the achievability of $\tau_p(M)$ for $p \geq \bar{p}$, almost all messages are assigned to the two transmitters that can be connected to the designated receiver, i.e., if we let $S_K=\left\{i: {\cal T}_{i,K} = \left\{i-1,i\right\}\right\}$, then $\lim_{K \rightarrow \infty} \frac{|S_K|}{K} = 1$. 

We recall from~\cite{ElGamal-Annapureddy-Veeravalli-arXiv12} that for the case of no erasures, the average per user DoF equals $\frac{2M}{2M+1}$. We also note that following the same footsteps as in the proof of~\cite[Theorem $7$]{ElGamal-Annapureddy-Veeravalli-arXiv12}, we can show that for any message assignment strategy such that $\lim_{K \rightarrow \infty} \frac{|S_K|}{K} = 1$, the per user DoF for the case of no erasures is upper bounded by $\frac{2M-2}{2M-1}$; we do so by using Lemma~\ref{lem:dofouterbound} for each $K-$user channel with the set ${\cal A}$ defined such that the complement set $\bar{\cal A}=\{i:i\in[K], i=(2M-1)(j-1)+M, j\in{\bf Z}^+\}$.
\end{proof}

The condition of optimality identified in the proof of Theorem~\ref{thm:comp} for message assignment strategies at high probabilities of erasure suggest a new role for cooperation in dynamic interference networks. The availability of a message at more than one transmitter may not only be used to cancel its interference at other receivers, but also to increase the chances of connecting the message to its designated receiver, i.e., to maximize \emph{coverage}. This new role leads to three effects at high erasure probability.  The achieved DoF in the considered linear interference network becomes larger than that of $K$ parallel channels, in particular, $\lim_{p \rightarrow 1} \frac{\tau_p(M>1)}{1-p} = 2$. Secondly, as the effect of interference diminishes at high probabilities of erasures, all messages can simply be assigned to the two transmitters that may be connected to their designated receiver, and a simple interference avoidance scheme can be used in each network realization, as we show below in the scheme of Theorem~\ref{thm:mtwoic}. It follows that channel state information is no longer needed at transmitters for interference management, and only information about the slow changes in the network topology is needed to achieve the optimal average DoF. Finally, unlike the optimal scheme of~\cite[Theorem $4$]{ElGamal-Annapureddy-Veeravalli-arXiv12} for the case of no erasures, where some transmitters are always inactive, achieving the optimal DoF at high probabilities of erasure requires all transmitters to be used in at least one network realization.

We now restrict our attention to the case where $M=2$. Here, each message can be available at two transmitters, and transmitted jointly by both of them. We first study in Theorems~\ref{thm:mtwoicaware} and~\ref{thm:mtwoic} two message assignment strategies that are optimal in the limits of $p \rightarrow 0$ and $p \rightarrow 1$, respectively, and derive closed form expressions for inner bounds on the average per user DoF $\tau_p(M=2)$ based on the considered strategies.

In~\cite{ElGamal-Annapureddy-Veeravalli-arXiv12}, the message assignment of Figure~\ref{fig:mtwojone} was shown to be DoF optimal for the case of no erasures ($p=0$). The network is split into subnetworks, each with five consecutive users. The last transmitter of each subnetwork is deactivated to eliminate inter-subnetwork interference. In the first subnetwork, message $W_3$ is not transmitted, and each other message is received without interference at its designated receiver. Note that the transmit beams for messages $W_1$ and $W_5$ contributing to the transmit signals $X_2$ and $X_3$, respectively, are designed to cancel the interference at receivers $Y_2$ and $Y_4$, respectively. An analogous scheme is used in each following subnetwork. The value of $\tau_p(M=2)$ is thus $\frac{4}{5}$ for the case where $p=0$. In order to prove the following result, we extend the message assignment of Figure~\ref{fig:mtwojone} to consider the possible presence of block erasures. 

\begin{thm}\label{thm:mtwoicaware}
For $M=2$, the following average per user DoF is achievable using a zero-forcing scheme,
\begin{equation}\label{eq:mtwoicaware}
\tau_p^{(\text{ZF})}(M=2) \geq \frac{1}{5}(1-p)\left(4+A \cdot p\right), 
\end{equation}
where
\begin{equation}
A=2p+\left(1-(1-p)^2+p(1-p)^3\right)\left(1+(1-p)^2\right),
\end{equation}
and
\begin{equation}
\lim_{p \rightarrow 0} \tau_p(2)=\frac{4}{5}.
\end{equation}
\end{thm}
\begin{IEEEproof}
We know from~\cite{ElGamal-Annapureddy-Veeravalli-arXiv12} that $\lim_{p \rightarrow 0} \tau_p(2)=\frac{4}{5}$, and hence, it suffices to show that the inner bound in~\eqref{eq:mtwoicaware} is valid. For each $i \in [K]$, message $W_i$ is assigned as follows,

 \vspace{5 mm}
 ${\cal T}_{i}=
 \begin{cases}
 \{i-1,i\}, \quad &\text{ if } i \equiv 2 \text{ mod } 5, \text{ or } i \equiv 4 \text{ mod } 5,\\
 \{i-1,i-2\}, \quad &\text{ if } i \equiv 0 \text{ mod } 5,\\
 \{i,i+1\}, \quad &\text{ otherwise},
 \end{cases}$
 \vspace{5 mm}
 
 We illustrate this message assignment in Figure~\ref{fig:mtwojonenew}. We note that the transmit signals $\{X_i: i \equiv 0 \text{ mod } 5\}$ are inactive, and hence, we split the network into five user subnetworks with no interference between successive subnetworks. We explain the transmission scheme in the first subnetwork and note that a similar scheme applies to each following subnetwork. In the proposed transmission scheme, any receiver is either inactive or receives its desired message without interference, and any transmitter will not transmit more than one message for any network realization. It follows that $1$ DoF is achieved for each message that is transmitted.
 
Messages $W_1$, $W_2$, $W_4$, and $W_5$ are transmitted through $X_1$, $X_2$, $X_3$, and $X_4$, respectively, whenever the coefficients $H_{1,1}\neq 0$, $H_{2,2}\neq 0$, $H_{4,3}\neq 0$, and $H_{5,4}\neq 0$, respectively. Note that the transmit beam for message $W_1$ contributing to $X_2$ can be designed to cancel its interference at $Y_2$. Similarly, the interference caused by $W_5$ at $Y_4$ can be canceled through $X_3$. There is an extra case where $W_4$ should be delivered through $X_4$, if $H_{4,4}\neq 0$ and $H_{4,3} = 0$ and $H_{5,4}=0$. This extra case takes place with probability $p^2(1-p)$. Similarly, there is an extra case for delivering $W_2$ through $X_1$ if $H_{2,1} \neq 0$ and $H_{1,1}=0$ and $H_{2,2}=0$. It follows that $(1-p)$ DoF is achieved for each of $W_1$ and $W_5$ and $(1-p)(1+p^2)$ DoF is achieved for each of $W_2$ and $W_4$, and hence, $\tau_p(2) \geq \frac{4}{5} (1-p)+\frac{2}{5}p^2(1-p)$. 

Now, we consider possibilities for sending message $W_3$ without violating the above four priority links. First, message $W_3$ can only be delivered through $X_3$ if it is either the case that one of $H_{2,2}$ and $H_{3,2}$ is erased, or it is the case that $H_{1,1}=0$ and all of $H_{2,1}, H_{2,2}, H_{3,2}$ exist. Otherwise, $W_2$ will be transmitted through $X_2$ and will cause interference at $Y_3$ that cannot be eliminated. This event will take place with probability $\left(1-(1-p)^2+p(1-p)^3\right)$. Further, $H_{3,3}$ has to exist, which has probability $(1-p)$. Also, it is either the case that $H_{4,3} = 0$ so that $W_4$ cannot be delivered through $X_3$ and $W_3$ would not cause interference at $Y_4$, or it is the case that $H_{4,3} \neq 0, H_{4,4} \neq 0$ and $H_{5,4} = 0$, and hence $W_3$ can be delivered through $X_3$ and its interference is canceled at $Y_4$, and $W_4$ is delivered through $X_4$, while $W_5$ cannot be delivered through $X_4$. The first of these cases takes place with probability $p$, and the second takes place with probability $p(1-p)^2$.

Summing up the probabilities of the above transmissions, we get the inner bound in~\eqref{eq:mtwoicaware}.
\end{IEEEproof}

\begin{figure}
  \centering
  
\subfloat[]{\label{fig:mtwojone}\includegraphics[height=0.26\textwidth]{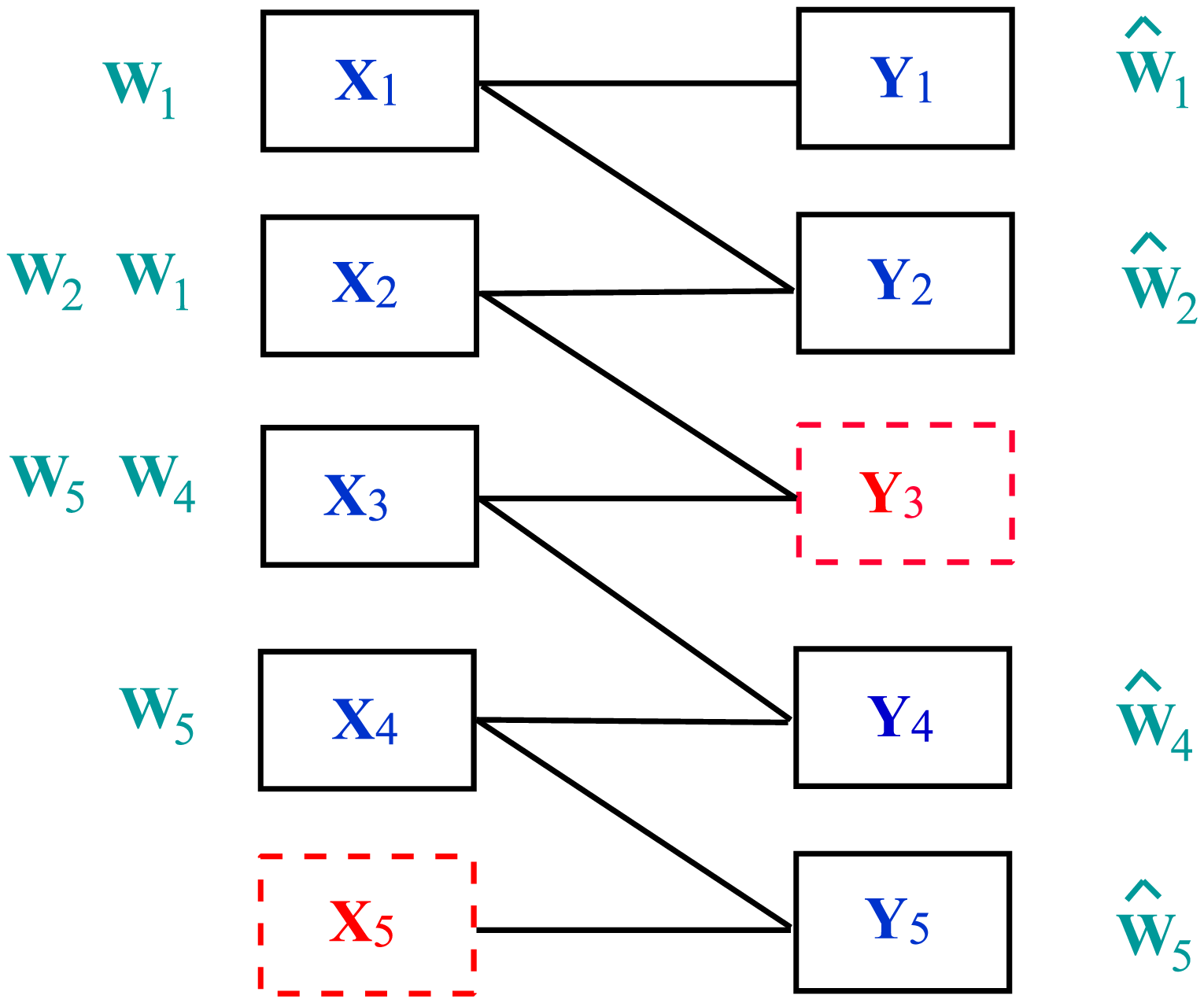}}                
\quad\quad\quad\quad\subfloat[]{\label{fig:mtwojonenew}\includegraphics[width=0.3\textwidth]{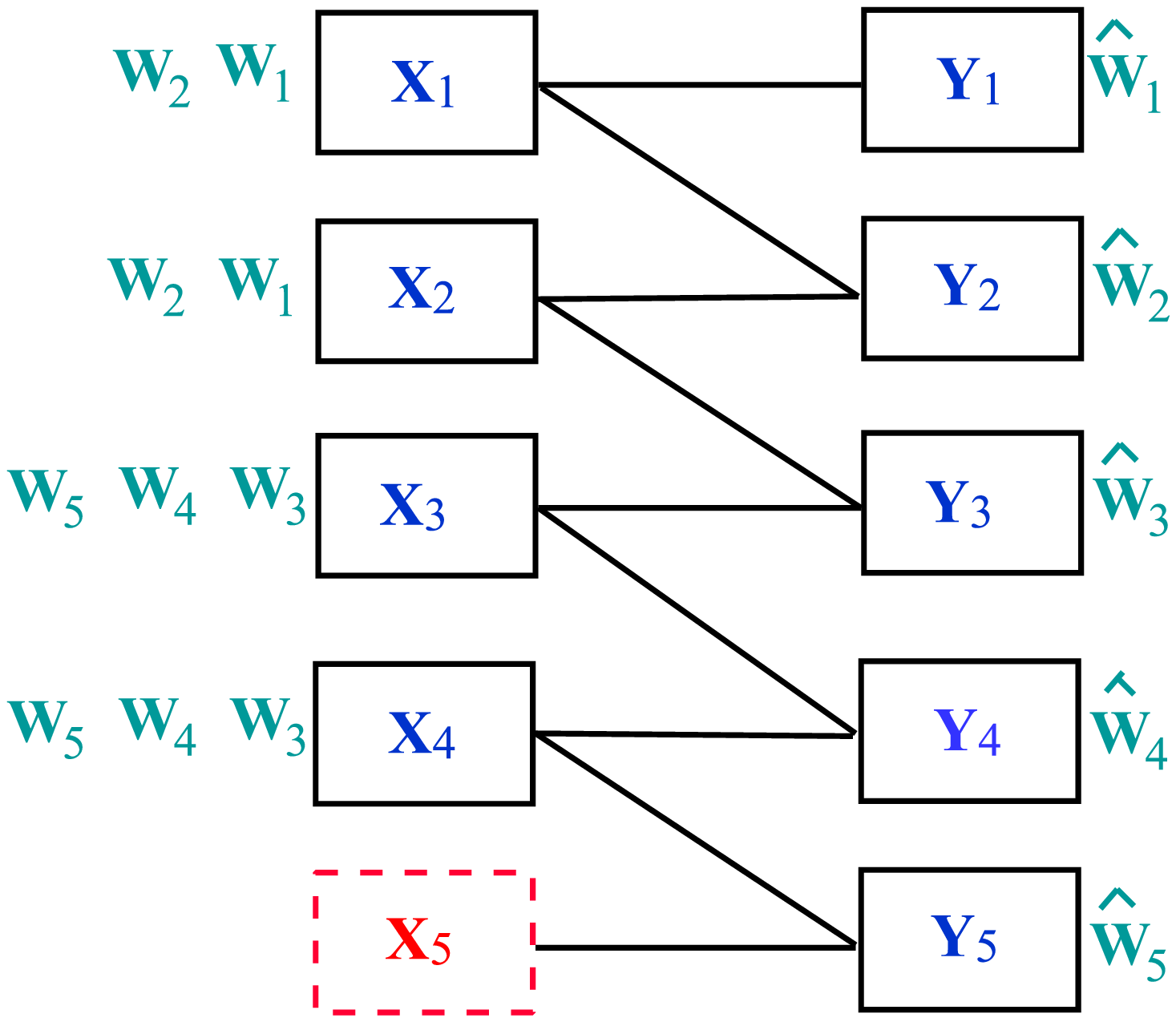}}
  \caption{The message assignment in ($a$) is optimal for a linear network with no erasures ($p=0$). We extend this message assignment in ($b$) to consider non-zero erasure probabilities. In both figures, the red dashed boxes correspond to inactive signals.}
  \label{fig:mtwoicaware}
\end{figure}

Although the scheme of Theorem~\ref{thm:mtwoicaware} is optimal for the case of no erasures ($p=0$), we know from Theorem~\ref{thm:comp} that better schemes exist at high erasure probabilities. Since in each five user subnetwork in the scheme of Theorem~\ref{thm:mtwoicaware}, only two users have their messages assigned to the two transmitters that can be connected to the destined receiver, and four users have only one of these transmitters carrying their messages, we get the asymptotic limit of  $\frac{7}{5}$ for the achieved average per user DoF normalized by $(1-p)$ as $p \rightarrow 1$. This leads us to consider an alternative message assignment where the two transmitters carrying each message $i$ are the two transmitters $\left\{i-1,i\right\}$ that can be connected to its designated receiver. Such assignment would lead the ratio $\frac{\tau_p(2)}{1-p} \rightarrow 2$ as $p \rightarrow 1$. In the following theorem, we analyze a transmission scheme based on this assignment.

\begin{thm}\label{thm:mtwoic}
For $M=2$, the following average per user DoF is achievable using a zero-forcing scheme,
\begin{equation}\label{eq:mtwoic}
\tau_p^{(\text{ZF})}(M=2) \geq \frac{1}{3} \left(1-p\right)\left(1 + \left(1-p\right)^3 + B \cdot p\right),
\end{equation}
where
\begin{eqnarray}
B &=& 3+\left(1+\left(1-p\right)^3\right)\left(1-\left(1-p\right)^2+p\left(1-p\right)^3\right)\nonumber\\ &+& p \left(1+(1-p)^2 \right),
\end{eqnarray}
and
\begin{equation}\label{eq:mtwoiclimit}
\lim_{p \rightarrow 1} \frac{\tau_p(2)}{1-p} = 2.
\end{equation}
\end{thm}
\begin{proof}
For any message assignment, no message can be transmitted if the links from both transmitters carrying the message to its designated receiver are absent, and hence, the average DoF achieved for each message is at most $1-p^2$. It follows that $\lim_{p\rightarrow 1} \frac{\tau_p(2)}{1-p} \leq \lim_{p \rightarrow 1} \frac{(1-p)(1+p)}{1-p} = 2$. We then need only to prove that the inner bound in~\eqref{eq:mtwoic} is valid. In the achieving scheme, each message is assigned to the two transmitters that may be connected to its designated receiver, i.e., ${\cal T}_i = \{i-1,i\}, \forall i\in[K]$. Also, in each network realization, each transmitter will transmit at most one message and any transmitted message will be received at its designated receiver without interference. It follows that 1 DoF is achieved for any message that is transmitted, and hence, the probability of transmission is the same as the average DoF achieved for each message.

Each message $W_i$ such that $i \equiv 0 \text{ mod } 3$ is transmitted through $X_{i-1}$ whenever $H_{i,i-1} \neq 0$, and is transmitted through $X_i$ whenever $H_{i,i-1}=0$ and $H_{i,i} \neq 0$. It follows that $d_0$ DoF is achieved for each of these messages, where,
\begin{equation}
d_0 = (1-p)(1+p).
\end{equation}

We now consider messages $W_i$ such that $i \equiv 1 \text{ mod } 3$. Any such message is transmitted through $X_{i-1}$ whenever $H_{i,i-1} \neq 0$ and $H_{i-1,i-1} = 0$. We note that whenever the channel coefficient $H_{i-1,i-1} \neq 0$, message $W_i$ cannot be transmitted through $X_{i-1}$ as the transmission of $W_i$ through $X_{i-1}$ in this case will prevent $W_{i-1}$ from being transmitted due to either interference at $Y_{i-1}$ or sharing the transmitter $X_{i-1}$. It follows that $d_1^{(1)} = p(1-p)$ DoF is achieved for transmission of $W_i$ through $X_{i-1}$. Also, message $W_i$ is transmitted through $X_i$ whenever it is not transmitted through $X_{i-1}$ and $H_{i,i} \neq 0$ and either $H_{i,i-1}=0$ or message $W_{i-1}$ is transmitted through $X_{i-2}$. More precisely, $W_i$ is transmitted through $X_i$ whenever all the following is true:
$H_{i,i} \neq 0$, and either $H_{i,i-1} =0$ or it is the case that $H_{i,i-1}\neq 0$ and $H_{i-1,i-1} \neq 0$ and $H_{i-1,i-2} \neq 0$.
It follows that $d_1^{(2)}=p\left(1-p\right)+\left(1-p\right)^4$ is achieved for transmission of $W_i$ through $X_{i}$, and hence, $d_1$ DoF is achieved for each message $W_i$ such that $i \equiv 1 \text{ mod } 3$, where,
\begin{equation}
d_1= d_1^{(1)} + d_1^{(2)} = 2p\left(1-p\right) + \left(1-p\right)^4.
\end{equation}

We now consider messages $W_i$ such that $i \equiv 2 \text{ mod } 3$. Any such message is transmitted through $X_{i-1}$ whenever all the following is true: 

\begin{itemize}
\item $H_{i,i-1} \neq 0$.
\item Either $H_{i-1,i-1}=0,$ or $W_{i-1}$ is not transmitted.
\item $W_{i+1}$ is not causing interference at $Y_i$.
\end{itemize} 

The first condition is satisfied with probability $(1-p)$. In order to compute the probability of satisfying the second condition, we note that $W_{i-1}$ is not transmitted for the case when $H_{i-1,i-1} \neq 0$ only if $W_{i-2}$ is transmitted through $X_{i-2}$ and causing interference at $Y_{i-1}$, i.e., only if $H_{i-2,i-3}=0$ and $H_{i-2,i-2} \neq 0$ and $H_{i-1,i-2} \neq 0$. It follows that the second condition is satisfied with probability $p + p(1-p)^3$. The third condition is not satisfied only if $H_{i,i} \neq 0$ and $H_{i+1,i} \neq 0$, and hence, will be satisfied with probability at least $1-\left(1-p\right)^2$. Moreover, even if $H_{i,i} \neq 0$ and $H_{i+1,i} \neq 0$, the third condition can be satisfied if message $W_{i+1}$ can be transmitted through $X_{i+1}$ without causing interference at $Y_{i+2}$, i.e., if $H_{i+1,i+1} \neq 0$ and $H_{i+2,i+1}=0$. It follows that the third condition will be satisfied with probability $1-(1-p)^2+p(1-p)^3$, and $d_2^{(1)}$ DoF is achieved by transmission of $W_i$ through $X_{i-1}$, where,
\begin{equation}
d_2^{(1)} = p\left(1-p\right)\left(1+\left(1-p\right)^3\right)\left(1-\left(1-p\right)^2+p\left(1-p\right)^3\right).
\end{equation}
 
Message $W_i$ such that $i \equiv 2 \text{ mod } 3$ is transmitted through $X_i$ whenever $H_{i,i} \neq 0$, and $H_{i+1,i} = 0$, and either $H_{i,i-1} = 0$ or $W_{i-1}$ is transmitted through $X_{i-2}$. It follows that $d_2^{(2)}$ DoF is achieved by transmission of $W_i$ through $X_i$, where,
\begin{eqnarray}
d_2^{(2)} &=& p\left(1-p\right)\left(p + d_1^{(1)} \left(1-p\right) \right)
\\&=& p^2\left(1-p\right) \left(1+(1-p)^2 \right),
\end{eqnarray}
and hence, $d_2= d_2^{(1)} + d_2^{(2)}$ DoF is achieved for each message $W_i$ such that $i \equiv 2 \text{ mod } 3$. We finally get,
\begin{equation}
\tau_p(2) \geq \frac{d_0 + d_1 + d_2}{3},
\end{equation}   
which is the same inequality as in~\eqref{eq:mtwoic}.
\end{proof}

\begin{figure}
  \centering
\subfloat[]{\label{fig:mtwodof}\includegraphics[height=0.25\textwidth]{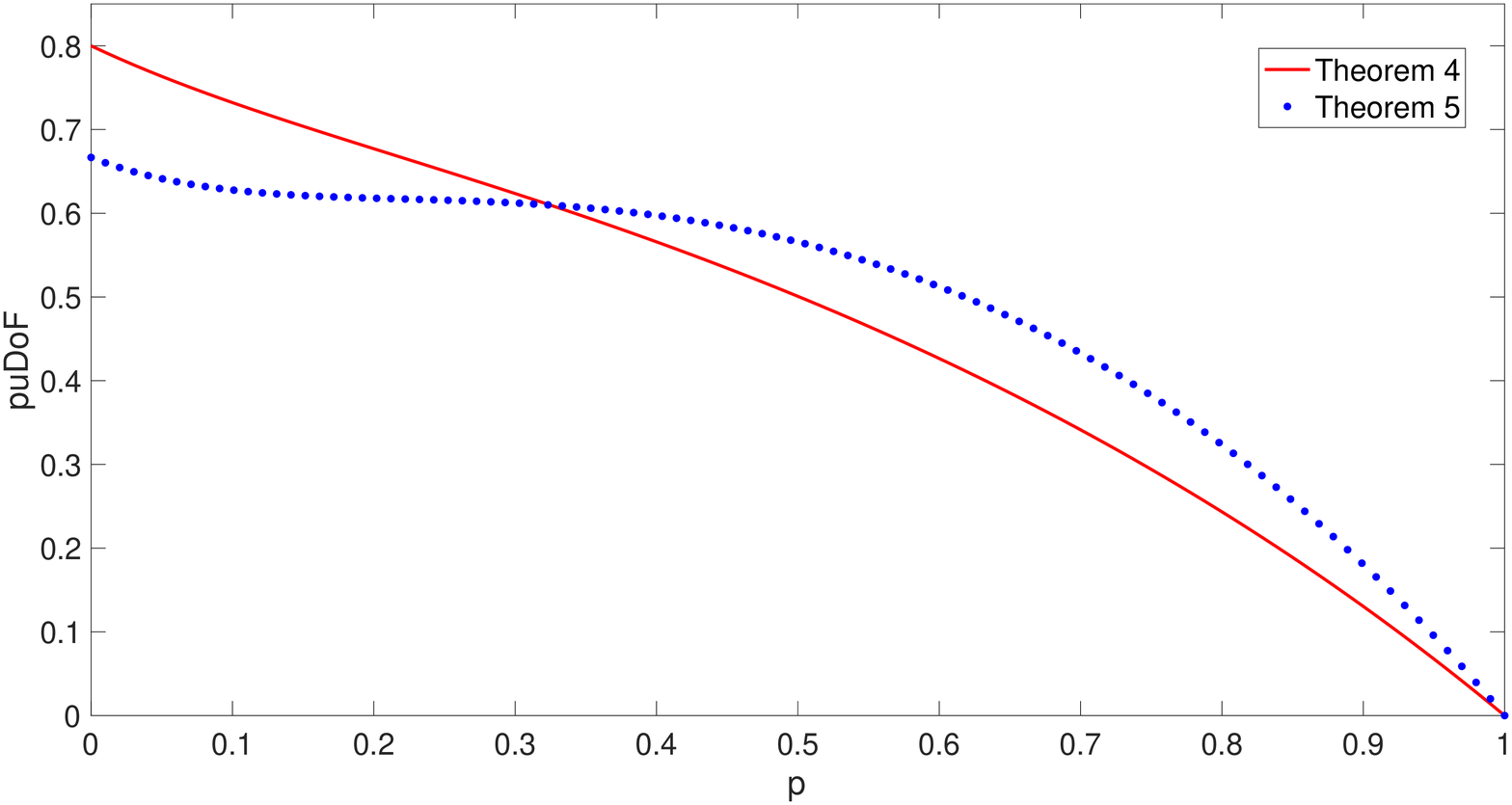}}                
\quad\quad\quad\quad\subfloat[]{\label{fig:mtwonormdof}\includegraphics[width=0.52\textwidth]{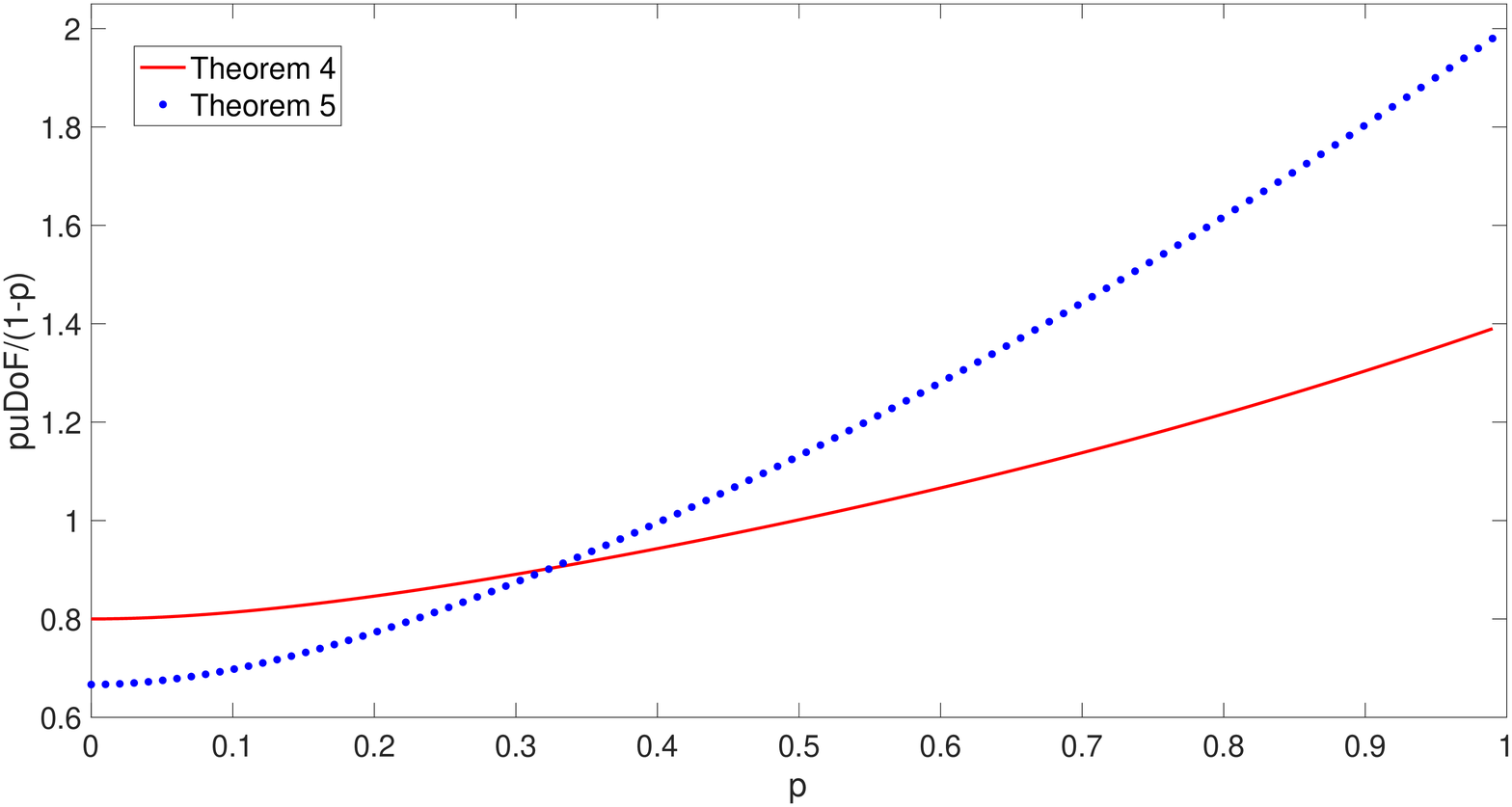}}
  \caption{Achieved inner bounds in Theorems $4$ and $5$. In $(a)$ we plot the achieved per user DoF. In $(b)$, we plot the achieved per user DoF normalized by $(1-p)$.}
  \label{fig:mtwo}
\end{figure}
We plot the inner bounds of~\eqref{eq:mtwoicaware} and~\eqref{eq:mtwoic} in Figure~\ref{fig:mtwo}. We note that below a threshold erasure probability $p \approx 0.34$, the scheme of Theorem~\ref{thm:mtwoicaware} is better, and hence  is proposed to be used in this case. For higher probabilities of erasure, the scheme of Theorem~\ref{thm:mtwoic} should be used. 
It is worth mentioning that we also studied a scheme based on the message assignment ${\cal T}_i = \{i,i+1\}, \forall i\in[K-1]$, that is introduced in~\cite{Lapidoth-Shamai-Wigger-ISIT07}. However, we did not include it here as it does not increase the maximum of the bounds derived in~\eqref{eq:mtwoicaware} and~\eqref{eq:mtwoic} at any value of $p$. Finally, although the considered channel model allows for using the interference alignment scheme of~\cite{Cadambe-IA} over multiple channel realizations (symbol extensions), all the proposed schemes require only coding over one channel realization because of the sparsity of the linear network.

\subsection{Optimal zero-forcing scheme with two transmitters per message ($M=2$)}
\label{sec:zero}

We have presented above two schemes that rely on zero-forcing transmit beamforming, and achieve the limits for $\tau_p(2)$ as $p \rightarrow 0$ and $p \rightarrow 1$. Here, we present an algorithm that reduces the problem of characterizing $\tau_p^{(\text{ZF})}(2)$ to that of identifying the optimal message assignment, at every value of $p$. 




We present Algorithm $1$ below that takes as input an atomic subnetwork with user indices $[N]=\{1,2,\cdots,N\}$, and outputs the transmit signals $\{X_i,i\in[N]\}$ which employs zero-forcing transmit beamforming to maximize the DoF value for users within the atomic subnetwork. Note that the first receiver can be connected to a transmitter with a preceding index; we let the index of that transmitter be zero.

We define a set of binary variables $b_{i,j}, j \in \{i-2,~i-1,~i,~i+1\}$ for each message $W_{i}$ that are initialized to zero. For each message we determine the conditions under which a message can be sent and decoded at its desired receiver, such that no interference occurs, in a successive order starting from  $W_1$ to $W_N$. For each decision to transmit a message $W_{i}$ from transmitter $j$, the corresponding variable $b_{i,j}$ is set to one.

In the algorithm below, two cases are considered for sending message $W_{i}$ to its destination by using either transmitter $i$ or $i-1$.  In the following we are discussing and justifying optimal choices in both cases. Due to their position at the beginning of the atomic subnetwork, the first two users represent a special case and thus are considered separately in lines $4-21$.  

\textit{Case 1}: Starting at  line $23$ of the algorithm, we examine the conditions for sending $W_{i}$ from transmitter $i-1$. Clearly this is only possible if $W_{i}$ is available at transmitter $i-1$. i.e.  $(i-1)\in {\cal T}_{i}$. We first restrict our attention to the case when transmitter $i-1$ does not send $W_{{i-1}}$. i.e., $b_{i-1,i-1}=0$. To ensure that transmitter $i-1$ does not cause interference at receiver $i-1$, we have to consider the possibility when receiver $i-1$ is not able to decode its desired message anyway, i.e. $W_{i-1}$ is not sent from transmitter $i-2$. If these above conditions are satisfied, then we set $b_{i,i-1}=1$. Further, starting at line $26$, we consider the scenario where $W_i$ would cause interference at receiver $i-1$, but this interference can be canceled by transmitting $W_i$ from transmitter $i-2$. In this case, we have to make sure that $W_i$ is available at transmitter $i-2$, i.e., $(i-2)\in{\cal T}_i$. Also, we have to ensure that transmitting $W_i$ from transmitter $i-2$ would not cause interference at receiver $i-2$. This would happen when $W_{i-2}$ is not being delivered, i.e., $b_{i-2,i-2}=0$ and $b_{i-2,i-3}=0$.
The last scenario starts at line $29$ and considers the case when both $W_i$ and $W_{i-1}$ can be delivered through $X_{i-1}$, while making sure that the interference they cause at receivers $i-1$ and $i$, respectively, is canceled.
In this last case, $W_i$ is transmitted from transmitters $i-2$ and $i-1$ and $W_{i-1}$ is transmitted from transmitters $i$ and $i+1$. Note that for any irreducible message assignment, $W_i$ can be available at transmitter $i-2$ only if it is available at transmitter $i-1$; this is why we do not check for $(i-1) \in {\cal T}_i$ at line $29$.

\textit{Case 2}: In the second case message $W_{i}$ is sent from transmitter $i$ (lines 32-37). The trivial conditions for achieving this are that message $W_{i}$ is available at transmitter $i$, and $W_{i}$ is not being delivered through transmitter $i-1$. Here, we also have to ensure that receiver $i$ can decode message $W_{i}$ without any interference. That is if transmitter $i-1$ is not active. Hence, we set $b_{i,i}$ to $1$ if the above conditions are satisfied. Further, there is on more case where the interference from transmitter $i-1$ can be canceled whenever message $W_{i-1}$ is available at transmitter $i$ and receiver $i$ only experiences interference by $W_{i-1}$. Consequently, both $b_{i,i}$ and $b_{i-1,i}$ are set to $1$ in this last case.

\begin{algorithm}\label{alg:zf}
\caption{Algorithm leads to optimal zero-forcing DoF value within an atomic subnetwork.}
  \begin{algorithmic}[1]
    \label{alg:the_alg}
\For{i=1:N}
\\ Define $b_{i,i-2}=b_{i,i-1}=b_{i,i}=b_{i,i+1}=0$
\EndFor
\If{$H_{1,0} \neq 0~ \land~0 \in \mathcal{T}_{1}$}  
 \State {$b_{1,0}=1$}
\Else
 \State {$b_{1,1}=1$}
\EndIf
\If{$1\in \mathcal{T}_{2}~\land~b_{1,1}=0$}
\If{$0\in \mathcal{T}_{2}~\land~ H_{1,0} \neq 0$}
\State$b_{2,1}=1$;  $b_{2,0}=1$
\EndIf
\ElsIf{${0\in \mathcal{T}_{2}}~\land~ H_{1,0} \neq 0~\land~2\in \mathcal{T}_{1}$} 
\State$b_{2,1}=1$,
{$b_{2,0}=1$,
$b_{1,2}=1$},
$b_{1,1}=1$
\ElsIf{$2 \in \mathcal{T}_{2}$}  
\If{$b_{1,1}=0$}  
\State$b_{2,2}=1$
\ElsIf {$2\in \mathcal{T}_{1}$}
\State$b_{2,2}=1$; $b_{1,2}=1$
\EndIf
\EndIf

\For{i=3:N}

\If{$(i-1)\in \mathcal{T}_{i}~\land~b_{i-1,i-1}=0$}
\If {$b_{i-1,i-2}=0$} 
\State{$b_{i,i-1} = 1$}
\ElsIf
 {$(i-2)\in \mathcal{T}_{i}~\land~b_{i-2,i-2}=0\land~b_{i-2, i-3}=0$}
\State$b_{i,i-1}=1$;  $b_{i,i-2}=1$
\EndIf

\ElsIf{${(i-2)\in \mathcal{T}_{i}}~\land~i\in \mathcal{T}_{i-1}\land b_{i-2,i-3}=0 \land b_{i-2,i-2}=0$} 

\State$b_{i,i-1}=1$,
{$b_{i,i-2}=1$,
$b_{i-1,i}=1$},
$b_{i-1,i-1}=1$
\EndIf
\If {$i \in \mathcal{T}_{i}\land b_{i,i-1}=0\land b_{i-2,i-1}=0$}
\If {$b_{i-1,i-1}=0$}
\State$b_{i,i}=1$

\ElsIf {$i \in \mathcal{T}_{i-1}$}
\State$b_{i,i}=1$;
$b_{i-1,i}=1$
\EndIf
\EndIf
\EndFor
	
\algstore{myalg}
\end{algorithmic}
\end{algorithm}
\begin{algorithm}[H]
	\setstretch{1.13}
	\begin{algorithmic} [1]              
		\algrestore{myalg}	
		
\For{i=0:N-1}\\
Set $X_{i} = 0$\\
 Generate $X_{i,i}$ from $W_j$ using an optimal AWGN channel point-to-point code (see e.g.,~\cite{Cover-Thomas})
 \If{$i > 0$}
 \State{Generate $X_{i,i-1}$ from $W_j$ using an optimal AWGN channel point-to-point code}
 \If{$b_{i,i}=1$}
 \State{$X_i \gets X_i + X_{i,i}$}
 \EndIf
 \EndIf
 \If{$b_{i+1,i}=1$}
 \State{$X_i \gets X_i + X_{i,i+1}$}
 \EndIf
 \EndFor
 \If{$H_{N,N}\neq 0$} 
 \State{Set $X_N = 0$}
 \State{Generate $X_{N,j},~j \in \{N-1,N\}$ from $W_N$ using an optimal AWGN channel point-to-point code.}
\If{$b_{N,N}=1$}
 \State{$X_N \gets X_N + X_{N,N}$}
 \EndIf

 \EndIf

\For{i = 0:N}
 \If{$i \geq 2~\land~b_{i-1,i} = 1$}
 \State{$X_i \gets X_i - \frac{H_{i,i-1} X_{i-1,i-1}}{H_{i,i}}$}
 \EndIf
\If{$i \leq N-2~\land~b_{i+2,i}=1$}
\State{$X_i \gets X_i - \frac{H_{i+1,i+1} X_{i+2,i+1}}{H_{i+1,i}}$}
\EndIf
\EndFor
\end{algorithmic}

\end{algorithm}

\begin{lem}\label{lem:cluster}
	For any message assignment such that each message is only assigned to two transmitters ($M=2$), Algorithm $1$  leads to the DoF-optimal zero-forcing transmission scheme for users within the input atomic subnetwork.
	
\label{lem_m1}
\end{lem}

\begin{IEEEproof}
	
	We consider in the algorithm the messages in ascending order from $W_{1}$ to $W_{N}$, and check which transmitter can deliver message $W_i$ such that it can be decoded at its desired receiver without interfering at any previous active receiver. If this is true, the message is transmitted. Also, if this is possible through any of the transmitters $i$ and $i-1$, then there is priority to transmit $W_i$ from transmitter $i-1$. In the following, we prove by induction that this procedure leads to the optimal transmission scheme. We first consider the base case, i.e., we prove that transmitting $W_1$ from transmitter $0$ is always optimal if it is available and the link $H_{1,0} \neq 0$. Otherwise, transmitting $W_1$ from transmitter $1$ would be optimal. 
    

Consider the feasible set $\varOmega \triangleq \{H_{i,j}: i \in [N], j\in{\cal T}_i, H_{i,j} \neq 0\}$ that denotes the subset of all links $H_{i,j}$ through which a message $W_{i}$ can be sent and decoded at its desired receiver. 
Assume an arbitrary set of links $\mathcal{S} \subset \varOmega\backslash H_{1,0}$, such that all links in $\mathcal{S}$ can be used \textit{simultaneously} to deliver messages to their desired receivers while eliminating interference. 
If $H_{1,0} \in \varOmega$, we now show that we can either add $H_{1,0}$ to ${\cal S}$, or replace the first link in ${\cal S}$ by $H_{1,0}$ and illustrate how this replacement does not lead to a decrease in the DoF. First note that if $H_{1,0} \in \varOmega$, then no receiver in the atomic subnetwork would observe interference due to transmitting $W_1$ from transmitter $0$, since transmitter $0$ is only connected to receiver $1$. Further, if $H_{2,1}$ is the first link in ${\cal S}$, then we replace it by $H_{1,0}$ and the sum DoF in the atomic subnetwork would remain the same by delivering $W_1$ instead of $W_2$.

For the case when $H_{1,0} \notin \varOmega$, we apply the same technique as in the previous case by substituting $H_{1,0}$ with $H_{1,1}$. Again, the replacement of the first link in ${\cal S}$ by $H_{1,1}$ would not decrease the DoF due to the fact that sending $W_1$ from the first transmitter only causes interference at the second receiver, and since $H_{2,j},~j \in \{1,2\}$ is either not in $\mathcal{S}$ or it is the first link in $\mathcal{S}$ that is replaced by $H_{1,1}$, the transmission of $W_1$ from transmitter $1$ would not necessitate a decrease in the number of links in ${\cal S}$. Finally, when it is possible to deliver $W_1$ through either transmitter $0$ or transmitter $1$, then choosing transmitter $0$ can only reduce the interference caused by $W_1$ at subsequent receivers.
Hence, it is always optimal to transmit $W_1$ from the first transmitter in its transmit set $\mathcal{T}_{1}$ as long as the corresponding link exists.
	
Next, we extend the proof to all users by induction. The induction hypothesis in the $i^{\text{th}}$ step is as follows.
We assume that transmission decisions for messages $\{W_k: k < i\}$ have been taken optimally to maximize the sum DoF. Let ${\mathcal{S}}_{1}\subset \varOmega$ be the set of links $H_{k,l}$, through which a subset of the messages $\{W_k, k < i\}$ can be delivered simultaneously to their destinations, while eliminating interference. Assume that all links in ${\mathcal{S}}_{1}$ are chosen optimally, i.e. the number of delivered messages cannot be increased by changing any of these links. 


	
Then, we do the induction step. Let ${\mathcal{S}}_{2} \subset \varOmega$ be any set of links $H_{k,l}$, through which a subset of the messages $\{W_{k},k > i\}$ can be transmitted simultaneously such that they can be decoded at their destinations. Also, the links in ${\cal S}_2$ are chosen optimally to maximize the number of delivered messages. If it is possible to send $W_{i}$ through $H_{i,i-1}$ without causing a conflict with any of the messages, that are sent through the links in $\mathcal{S}_{1}$, the same logic applies to $H_{i,i-1}$ as to $H_{1,0}$ in the base case. More precisely, if $W_{i}$ does not interfere at any previous active receiver and it can be decoded at receiver $i$ while eliminating interference, $H_{i,i-1}$ can be either added to ${\mathcal{S}}_{2}$ or replace the first link in ${\mathcal{S}}_{2}$, in order to obtain an optimal set of links for the transmission of the messages $\{W_{k},k \geq i\}$. This is possible since again, $W_{i}$ does not cause interference at any active receiver with an index $k > i$, and any of the links $\{H_{i+1,k}, k\in\{i,i+1\}\}$ is either not in ${\mathcal{S}}_{2}$ or it is the link that is replaced by $H_{i,i-1}$. If it is not possible to send $W_i$ through $H_{i,i-1}$ without causing a conflict with any of the messages that are sent through the links in ${\cal S}_1$, but it is possible to do so through $H_{i,i}$, then again the same argument applies for adding $H_{i,i}$ to ${\cal S}_2$. Further, we note that the preference to send $W_i$ through $H_{i,i-1}$ is optimal, since $H_{i,i-1}$ may only cause a conflict with $H_{i+1,i}$ in ${\cal S}_2$, while $H_{i,i}$ may cause a conflict with any of $H_{i+1,i}$ and $H_{i+1,i+1}$.
Therefore, as long as the aforementioned preference rule is applied, sending a message $W_{i}$ through a link $H_{i,j}$ is always optimal as long as it is possible to decode $W_{i}$ at receiver $i$ without causing interference at a previous active receiver. 

We have hence shown that the greedy approach followed by Algorithm $1$ to first explore all possibilities to deliver $W_i$ through $H_{i,i-1}$, and if not possible, investigate all possibilities to deliver it through $H_{i,i}$, without interfering with any previous actively delivered message, is DoF-optimal under restriction to zero-forcing schemes.
\end{IEEEproof}
	
	The optimality of the greedy approach illustrated in the above proof simplifies the optimal algorithm in two ways. On the one hand, we can go through the links one by one and check if it is possible to send a message to its desired receiver without interfering with any of the previous active messages. If it is possible, we will always decide to send the message. On the other hand, decisions that we already made do not have to be changed later, because at each step we make sure to avoid conflicts with previously activated messages. This procedure is applied in Algorithm 1, as we illustrate below.

In the following, we derive the decision conditions for the first two messages in the input  atomic subnetwork. 
If $H_{1,0} \in \varOmega$, sending $W_{1}$ is optimal, as shown in the base case of the proof by induction. Hence, set $b_{1,0} = 1$. If not, then it must be the case that $H_{1,1} \in \varOmega$, because otherwise receiver $1$ would not have belonged to the atomic subnetwork. In that last case, we set $b_{1,1}=1$, since it is optimal then to send $W_{1}$ from transmitter $1$ as shown in the above proof.  


We next consider the possibilities for delivering $W_2$ to its destination through transmitter $1$. If $H_{2,1} \in \varOmega$, we have the following possibilities. If $0\in\mathcal{T}_{2}$ and $b_{1,0}=1$, then we can send $W_{2}$ from transmitter $1$ and cancel interference at the first receiver by sending $W_{2}$ from transmitter $0$. Hence, we set $b_{2,0}=1$ and $b_{2,1}=1$ in this first case. In the second case, we deliver both $W_1$ and $W_2$ to their destinations through transmitter $1$, and cancel their interference through transmitters $2$ and $0$, respectively. This second case is possible when ${\cal T}_1=\{1,2\}$ and ${\cal T}_2=\{0,1\}$ and transmitter $0$ is in the atomic subnetwork, i.e., $H_{1,0} \neq 0$. The above two possibilities are the only ones that exist for delivering $W_2$ through transmitter $1$. 


We next consider the cases for delivering $W_2$ through transmitter $2$. If $H_{2,2} \in \varOmega$ and we are not sending $W_2$ from the first transmitter, i.e. $b_{2,1}=0$, then if $b_{1,1}=0$, $W_1$ is not causing interference at the second receiver and we set $b_{2,2}=1$. The second possible case is when $W_1$ is causing interference at receiver $2$, but this interference can be canceled through transmitter $2$, i.e., when ${\cal T}_1=\{1,2\}$. In this case, we set $b_{2,2}=1$ and $b_{1,2}=1$.

The illustration of the greedy approach for transmitting messages $\{W_i, i\in\{3,4,\cdots,N\}\}$ follows from the explanation of lines $22-39$ of the algorithm, that is provided above before Lemma~\ref{lem:cluster}. We now show that Algorithm $1$ can be used to achieve the optimal zero-forcing DoF in a general $K$-user network.

\begin{thm}\label{thm:zf}
Algorithm $1$ can be used to achieve the optimal zero-forcing DoF for any message assignment satisfying the cooperation order constraint $M=2$, and any realization of a general $K$-user dynamic linear network.
\end{thm}
\begin{IEEEproof}
Consider any realization of a $K$-user linear dynamic network.  We show below how the network can be partitioned into atomic subnetworks with no inter-subnetwork interference. It then follows by Lemma~\ref{lem:cluster} that Algorithm $1$ achieves the optimal zero-forcing DoF in each atomic subnetwork. Further, since there is no interference between the subnetworks, and no transmitter in a subnetwork is connected to a receiver in another subnetwork, it follows that invoking Algorithm $1$ for each of the atomic subnetworks in the partition leads to the optimal zero-forcing DoF for the entire network.

We first form a grouping of non-erased channel links, such that each group consists of a maximal set of consecutive non-erased links. More precisely, we scan the links in ascending order of index, and check if they are erased, e.g., we first check $H_{1,0}$, then $H_{1,1}$, then $H_{2,1}$, then $H_{2,2}$, and so on. We start adding links to the first group until we encounter the first erased link, and then start adding following non-erased links to the second group until we encounter an erased link again, and so on. 

We then show how to alter the above grouping of links to form atomic subnetworks. For each group of links we do the following. We first scan the receivers connected to any link in the group in ascending order of index. We add the scanned receiver to a subnetwork if the message corresponding to the scanned receiver is available at a transmitter connected to that receiver. Otherwise, we end the current subnetwork, and start a new subnetwork, and resume scanning receivers. Now, we have a partitioning of the network into non-interfering subnetworks. We finally do the following to make sure the subnetworks are atomic. In each subnetwork, we scan the transmitters connected to receivers in the subnetwork in ascending order of index. If the scanned transmitter does not carry a message for a receiver in its subnetwork, then we split the subnetwork at the index of that transmitter. In other words, the two receivers connected to that transmitter will belong to two different subnetworks. We keep repeating the above process of scanning transmitters until we have that in each subnetwork, each transmitter connected to a receiver in the subnetwork is carrying at least one message for a receiver in the subnetwork. 

The above process shows how the network can be partitioned into atomic subnetworks. The optimality of Algorithm $1$ then follows from Lemma~\ref{lem:cluster} by applying the algorithm for each atomic subnetwork.

\end{IEEEproof}

\subsection{Information-theoretic optimality of Algorithm $1$}

In this section, we establish the optimality of Algorithm $1$ to characterize $\tau_p(M=2)$. Following the same footsteps as the proof of Theorem~\ref{thm:zf}, all we need is to prove that the algorithm leads to the optimal DoF within the input atomic subnetwork. We hence have the result. 
\begin{lem}\label{lem:n5}
For any message assignment such that each message is only assigned to two transmitters ($M=2$), Algorithm $1$  leads to the DoF-optimal transmission scheme for users within an input atomic subnetwork whose size $N \in \{1,2,3,4,5\}$.
\end{lem}

\begin{algorithm}\label{alg:convproof}
\caption{Using Lemma~\ref{lem:dofouterbound} to prove optimality of Algorithm $1$ with subnetwork size $N=5$.}
  \begin{algorithmic}[1]
    \label{alg:the_alg5}
\If{transmitter $0$ is in the atomic subnetwork and transmitter $5$ is not in the atomic subnetwork}  
	\If{$3\not\in \mathcal{T}_{5}$}  
		\If{$0 \not\in \mathcal{T}_{1}$}  
			\State{$\mathcal{A}=\{2,3,4\}$}
		\ElsIf{$3 \not\in \mathcal{T}_{3}$}  
			\State{$\mathcal{A}=\{1,2,4\}$}
		\ElsIf{$\mathcal{T}_{2}=\{1,2\}$}  
			\State{$\mathcal{A}=\{1,3,4\}$}
		\Else{}
			\State{$\mathcal{A}=\{1,2,4,5\}$}
		\EndIf		
	\ElsIf{$3\not\in \mathcal{T}_{4}$}  
		\If{$0 \not\in \mathcal{T}_{1}$}  
			\State{$\mathcal{A}=\{2,3,5\}$}
		\ElsIf{$3 \not\in \mathcal{T}_{3}$}  
			\State{$\mathcal{A}=\{1,2,5\}$}
		\ElsIf{$\mathcal{T}_{2}=\{1,2\}$}  
			\State{$\mathcal{A}=\{1,3,5\}$}
		\Else{}
			\State{$\mathcal{A}=\{1,2,4,5\}$}
		\EndIf
	\Else{}  
		\State{$\mathcal{A}=\{1,2,4,5\}$}
	\EndIf
\EndIf

\If{transmitter $0$ is not in the atomic subnetwork and transmitter $5$ is in the atomic subnetwork}  
	\If{$(2 \not\in \mathcal{T}_{1})$}
		\If{$5 \not\in \mathcal{T}_{5}$}
			\State{$\mathcal{A}=\{2,3,4\}$}
		\ElsIf{$2 \not\in \mathcal{T}_{3}$}
			\State{$\mathcal{A}=\{2,4,5\}$}
		\ElsIf{$\mathcal{T}_{4}=\{3,4\}$}
			\State{$\mathcal{A}=\{2,3,5\}$}
		\Else{}
	  		\State{$\mathcal{A}=\{1,2,4,5\}$}
		\EndIf{}
	\ElsIf{$(2 \not\in \mathcal{T}_{2})$}
		\If{$5 \not\in \mathcal{T}_{5}$}
			\State{$\mathcal{A}=\{1,3,4\}$}
		\ElsIf{$2 \not\in \mathcal{T}_{3}$}
			\State{$\mathcal{A}=\{1,4,5\}$}
		\ElsIf{$\mathcal{T}_{4}=\{3,4\}$}
			\State{$\mathcal{A}=\{1,3,5\}$}
		\Else{}
	  		\State{$\mathcal{A}=\{1,2,4,5\}$}
		\EndIf{}
	\Else{}
		\State{$\mathcal{A}=\{1,2,4,5\}$}
	\algstore{myalg}
\end{algorithmic}
\end{algorithm}
\begin{algorithm}[H]
	\setstretch{1.13}
	\begin{algorithmic} [1]              
		\algrestore{myalg}
    \EndIf{}
\EndIf
\If{both transmitters $0$ and $5$ are in the subnetwork}  
	\State{$\mathcal{A}=\{1,2,4,5\}$}
\EndIf
\If{both transmitters $0$ and $5$ are not in the subnetwork}  
	\If{$2 \not\in \mathcal{T}_{1}$}  
		\State{$\mathcal{A}=\{2,3,4\}$}
	\ElsIf{$2 \not\in \mathcal{T}_{2}$}  
		\State{$\mathcal{A}=\{1,3,4\}$}
	\ElsIf{$3 \not\in \mathcal{T}_{4}$}  
			\State{$\mathcal{A}=\{2,3,5\}$}
	 \ElsIf{$3 \not\in \mathcal{T}_{5}$}  
			\State{$\mathcal{A}=\{2,3,4\}$}
	\Else{}
			\State{ $\mathcal{A}=\{1,2,4,5\}$}	
		\EndIf		
	\EndIf

\end{algorithmic}
\end{algorithm}

\begin{IEEEproof}
We have formal proofs for the information-theoretic optimality of Algorithm $1$ for atomic subnetworks that have sizes $N \leq 5$. For brevity, we only state here the proof for $N=5$, as the proofs with smaller subnetwork sizes are simpler. The proof is based on exploring all possibilities for the message assignment, and establishing the optimality of Algorithm $1$ for each by showing that the DoF achieved by the algorithm is optimal. We do so through using Lemma $5$ by following the procedure stated in Algorithm $2$ to construct the set ${\cal A}$ of received signals that suffice to reconstruct all received signals in the subnetwork, with an uncertainty that does not increase with the transmit power $P$.

In what follows, we illustrate the steps followed by Algorithm $2$. Note that we say that \textbf{a transmitter is in the atomic subnetwork} if the following conditions are satisfied:
\begin{enumerate}
\item The transmitter is connected to a receiver whose index is in the subnetwork.
\item The transmitter is carrying a message whose index is in the subnetwork.
\end{enumerate}
Since the considered atomic subnetwork has size $N=5$, we know that transmitters with indices in the set $\{1,2,3,4\}$ are in the subnetwork. We hence consider cases on the membership of transmitters $0$ and $5$. 

The case where transmitter $0$ belongs to the subnetwork and transmitter $5$ does not is considered at the start of the algorithm. Note that Algorithm $1$ can always lead to achieving $3$ DoF in this case, since $W_5$ can be delivered through $X_4$, and $W_1$ and $W_2$ can be delivered simultaneously while eliminating interference. First consider the case when $W_5$ is not available at transmitter $3$:
\begin{enumerate}
\item If it is also the case that $W_1$ is not available at transmitter $0$, then we can reconstruct all transmit signals from $Y_2$, $Y_3$ and $Y_4$ as stated in line $4$ of the algorithm. This is because for any reliable communication scheme, we can reconstruct $W_2$, $W_3$ and $W_4$ from these received signals, and hence reconstruct $X_0$ and $X_3$ from the considered conditions on the message assignment. Following the linear connectivity of the network, we can then reconstruct $X_4, X_2$ and $X_1$ from $Y_4, Y_3$ and $Y_2$, with respect to order. 

\item The second case is when $W_3$ is not available at transmitter $3$, then we apply Lemma~\ref{lem:dofouterbound} with the set ${\cal A}=\{1,2,4\}$ as in line $6$ of the algorithm. Since both $W_3$ and $W_5$ do not contribute to $X_3$, we can reconstruct this transmit signal. We can then reconstruct $X_4$ from $Y_4$. Further, since transmitter $0$ can only have $W_1$ and $W_2$ since we can restrict our attention to irreducible message assignments (see~\cite[Chapter $6$]{Veeravalli-ElGamal-Cambridge}), and hence we can reconstruct $X_0$. Following the connectivity of the network, we can then reconstruct $X_1$ and $X_2$ from $Y_1$ and $Y_2$, respectively.

\item The third case is when $W_2$ is available at transmitters $1$ and $2$. In this case, we apply Lemma~\ref{lem:dofouterbound} with the set ${\cal A}=\{1,3,4\}$ as in line $8$ of Algorithm $2$. $X_0$ can be reconstructed since we know $W_1$, and $W_2$ is not available at transmitter $0$. Also, $X_1$ can then be reconstructed from $Y_1$. Since $W_2$ is not available at transmitters $3$, $4$, we can reconstruct $X_3$ and $X_4$ as well. Finally, $X_2$ can be reconstructed from $Y_3$.

\item The final case is when none of the above three cases applies. In this case, Lemma~\ref{lem:dofouterbound} applies with the set ${\cal A}=\{1,2,4,5\}$ per line $10$ of the algorithm. The upper bound proof here is the same as the case where no erasures occur~\cite[Chapter $6$]{Veeravalli-ElGamal-Cambridge}. Algorithm $1$ leads to achieving $4$ DoF within the input subnetwork in this case as follows. Since the first case above does not apply, then $W_1$ can be delivered through $X_0$. Since the second case above does not apply, then $W_3$ can be delivered through $X_3$. Also, since transmitter $5$ is not in the subnetwork, and $W_5$ is not available at transmitter $3$, then it has to be the case that $W_5$ is available at transmitter $4$, and can be delivered through $X_4$. Finally, since the second case above does not apply, then it is either the case that ${\cal T}_2=\{0,1\}$ and in this case $W_2$ is delivered through $X_1$ and its interference at $Y_1$ is canceled through $X_0$, or it is the case that ${\cal T}_2=\{2,3\}$ and in this case $W_2$ is delivered through $X_2$ and its interference at $Y_3$ is canceled through $X_3$.
\end{enumerate}
Note that the justification for the case when $W_5$ is available at transmitter $3$, but $W_4$ is not, which is investigated in lines $12-24$ of Algorithm $2$, is identical to the above case when $W_5$ is not available at transmitter $3$, but with switching the receiver and message indices $4$ and $5$. Further, the justification for the case when transmitter $5$ is in the subnetwork, but transmitter $0$ is not, which is investigated in lines $26-50$ of Algorithm $2$, is identical to that when transmitter $0$ is in the subnetwork, but transmitter $5$ is not, but with switching transmitter indices $i$ and $5-i$ and switching receiver indices $i$ and $6-i$ for $i \in \{1,2,3,4,5\}$. In other words, we view the subnetwork in this case as a mirrored version of the first case (upside down) and apply the same logic. It hence remains to check the cases when transmitters $0$ and $5$ are either both in the subnetwork or both not in. If both are in the subnetwork, then we apply Lemma~\ref{lem:dofouterbound} with ${\cal A}=\{1,2,4,5\}$ and the proof is the same as for the case of no erasures. Further, Algorithm $1$ in this case can be used to achieve $4$ DoF, as it is always possible to deliver $W_1, W_2, W_4$ and $W_5$ with no interference.

The final scenario is considered starting at line $54$ of Algorithm $2$, when both transmitters $0$ and $5$ are not in the subnetwork. Here also, Algorithm $1$ leads to achieving $3$ DoF within the input subnetwork, by delivering $W_1$ through $X_1$ and delivering $W_5$ through $X_4$, and delivering $W_3$ through either $X_2$ or $X_3$. We have the following cases for the converse proof.
\begin{enumerate}
\item The first case is when $W_1$ is not available at transmitter $2$, then we apply Lemma~\ref{lem:dofouterbound} with the set ${\cal A}=\{2,3,4\}$ as in line $56$ of Algorithm $2$. $X_2$ can be reconstructed, since $W_1$ and $W_5$ do not contribute to it. Further, the linear connectivity implies that we can reconstruct $X_1, X_3$ and $X_4$ from $Y_2, Y_3$ and $Y_4$, respectively.

\item The second case is when $W_2$ is not available at transmitter $2$, then we apply Lemma~\ref{lem:dofouterbound} with the set ${\cal A}=\{1,3,4\}$. In this case, $X_2$ can be reconstructed, since $W_2$ and $W_5$ do not contribute to it. Further, $X_1, X_3$ and $X_4$ can be reconstructed from $Y_1, Y_3$ and $Y_4$, respectively.
\end{enumerate}
The proof for the next two cases when $W_3$ is not available at transmitters $5$ and $4$, is identical to the above two cases, with respect to order, but with switching transmitter indices $i$ and $5-i$ and receiver and message indices $i$ and $6-i$ for $i\in\{1,2,3,4\}$. The remaining final case considered at line $63$ of Algorithm $2$, is when transmitter $2$ has both $W_1$ and $W_2$, and transmitter $3$ has both $W_4$ and $W_5$. In this case, Algorithm $1$ delivers $W_1$ through $X_1$, and its interference at $Y_2$ is canceled through $X_2$. Also, $W_2$ is delivered through $X_2$. Further, $W_5$ is delivered through $X_4$, and its interference at $Y_4$ is canceled through $X_3$. Finally, $W_4$ is delivered through $X_3$. The converse argument in this last case is the same as in the no erasure case.
\end{IEEEproof}
Although Lemma~\ref{lem:n5} applies only to realizations of the dynamic linear network, where the maximum size of an atomic subnetwork is $N=5$, it represents an important step towards understanding the optimal message assignment and transmission schedule for CoMP transmission as we note the following.
\begin{remk}
In order to have an atomic subnetwork whose size $N > 5$, we have to have at least $10$ consecutive channel links that are not erased, which takes place with a probability proportional to $(1-p)^{10}$. These events are extremely rare for significant erasure probabilities.
\end{remk}
\begin{remk}
We have not found any example of a message assignment satisfying the cooperation order constraint $M=2$ and an atomic subnetwork whose size $N>5$, where Lemma~\ref{lem:dofouterbound} cannot be used to prove the information-theoretic optimality of Algorithm $1$, in a manner similar to the proof of Lemma~\ref{lem:n5} using Algorithm $2$. Hence, we conjecture that Algorithm $1$ can in fact be used to characterize $\tau_p(M=2)$.
\end{remk}


\section{Simulation}



In this section, we investigate through simulations the best average per user DoF achievable through Algorithm $1$. We note that if our above conjecture that Lemma~\ref{lem:n5} extends to atomic subnetworks of arbitrarily large sizes, then what we are characterizing in this section is essentially $\tau_p(M=2)$. 
For a given message assignment and a certain erasure probability $p$, we compute the average puDoF by simulating\footnote{The MATLAB code is available at https://github.com/toluhatake/Fundamental-Limits-of-Dynamic-Interference-Management-with-Flexible-Message-Assignments} a sufficiently large number $n$ of channel realizations. Each link is erased with probability $p$ and the network is partitioned into atomic subnetworks before applying Algorithm $1$.  
The value of the average per user DoF is then computed as the average number of decoded messages divided by the network size $K$. 

To guarantee the validity of the computed value for large networks, we deactivate the last transmitter in the network. Hereby we ensure that for a large network that consists of concatenated subnetworks; each of size $K$, then the computed average per user DoF value can be achieved in the large network by repeating the scheme for each subnetwork, since there will be no inter-subnetwork interference.
		
The simulation is done for a set of message assignments with different fractions $f(p)$ of messages that are assigned to one transmitter connected to their desired receiver and another transmitter that can be used to cancel interference, while the remaining fraction of $1-f(p)$ of messages are assigned to both transmitters that are connected to their destination. Furthermore, we vary the network size $K$ to consider up to $100$ users. 
More precisely, we use the following assignment strategy.
		\begin{equation*}
		{\scriptsize
			\mathcal{T}_i = 
			\begin{cases}
			\{0, ~1\} & i = 1 ~\text{and}~$f(p) = 0.01$ ,\\
            \{1, ~2\} & i = 1 ~\text{and}~$f(p) > 0.01$ ,\\
			\{K-2,~K-1\} & i = K,\\
			\{i, ~i+1\} & i= 1+ n \cdot \max\left\{2,~\Bigl \lfloor \frac{K}{f(p)\cdot K -1}\Bigr \rfloor\right\},\\
			~& {\tiny n \in \{1,2,\dots, \min\left\{f(p) \cdot K -2,~ \Bigl\lfloor \frac{K}{2} - 1 \Bigr \rfloor\right\}},\\
			\{i, ~i+1\} & i= 2n ,~{\tiny n \in \left\{1,2,\dots,\Bigl\lceil (f(p) - \frac{1}{2} ) K \Bigr\rceil -1\right\}},\\ 
			\{i-1,~ i\} & \text{otherwise,}
			\end{cases}}
			\label{equation:eq1}
		\end{equation*} 	
		where we use the notation $\{1,2, ...,x\}$ to denote the set $[x]$ when $x$ $\geq$1 and the empty set when $x<1$.
	           We vary the value of $f(p)$ from $0$ up to $1$ in steps of $\frac{1}{100}$, calculating the average puDoF as a function of $p$ for each of these message assignments.

As a result, the maximum puDoF that is achievable with the set of message assignments described above is shown in Figure \ref{fig:Plot}. Compared to the schemes presented in Theorems~\ref{thm:mtwoicaware} and \ref{thm:mtwoic}, there exist message assignments with a better performance for middle ranges of $p$. These are presented in Table \ref{opt_assignments}. Note that in \cite{ElGamal-Veeravalli-Asilomar13}, it was shown that an assignment with $f(p) = \frac{2}{5}$ is optimal for $p \rightarrow 0$. Interestingly, we find the assignment presented in Theorem~\ref{thm:mtwoicaware} with $f(p) = \frac{3}{5}$ (see the green curve in Fig. \ref{fig:Plot}) achieves the same puDoF for $p = 0$, but performs slightly better on the interval $(0, 0.15]$. From our results in Table \ref{opt_assignments}, we observe that the optimal fraction $f(p)$ decreases monotonically from $\frac{3}{5}$ to $0$ as $p$ goes from $0$ to $1$, which has an intuitive explanation as the role of cooperation shifts from interference management, which requires a high value of $f(p)$, to increasing the coverage probability for each message, which requires a low value of $f(p)$, as the erasure probability increases. 
		
		\begin{figure}[H]	
		\centering
			\includegraphics[width = \columnwidth]{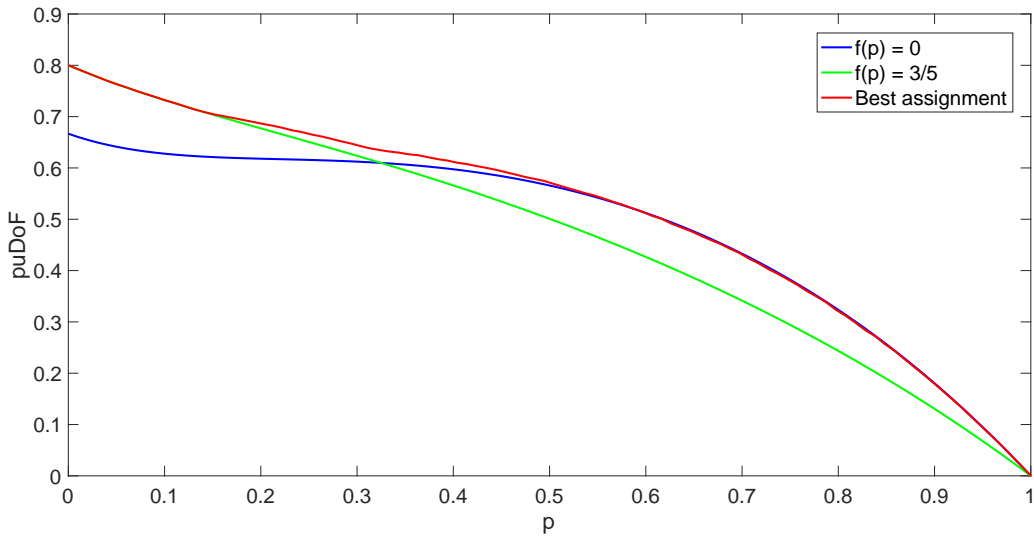}			\caption{The plot shows the puDoF as a function of the erasure probability $p$ found by applying Algorithm $1$ to $6000$ randomly generated channel realizations for each value of $p\in\{0,0.01,0.02,\cdots,1\}$. The green and blue curves correspond to the message assignment and transmission strategy presented in Theorems~\ref{thm:mtwoicaware} and \ref{thm:mtwoic}, which were shown to be optimal as $p \rightarrow 0$, and $p \rightarrow 1$, with respect to order. The red curve is the maximum puDoF that is achievable with the message assignments we considered in our simulation.}
			\label{fig:Plot}
		\end{figure}
			
		\begin{table}[h]
			\begin{center}
			\begin{tabular}{|c|c|}
				\hline 
				Range of $p$ & Value of $f(p)$ for best performing message assignment \\ 
				\hline \hline
				0 to 0.15 & $\frac{3}{5}$ (as in Theorem~\ref{thm:mtwoicaware})  \\ 
				\hline 
				0.16 to 0.29 & $\frac{1}{2}$ \\ 
				\hline 
				0.3  & $\frac{49}{100}$ \\ 
				\hline 
				0.31 to 0.32 & $\frac{12}{25}$ \\ 
				\hline 
				0.33 to 0.58 & $\frac{1}{50}$ \\ 
				\hline 
				0.59 to 1 & $0$ (as in Theorem~\ref{thm:mtwoic}) \\ 
				\hline 
			\end{tabular} 
			\end{center}
			\caption{Message assignments with the best performance out of the set of assignments that was simulated.}
			\label{opt_assignments}
		\end{table}

In Figure~\ref{fig:CoMP}, we plot the value of $\tau_p($M=1$)$ derived in Section~\ref{sec:cellassociation} versus the best average puDoF value obtained in this simulation for $M=2$. One could observe from the simulation the value of the extra backhaul budget at each value of the erasure probability $p$. Interestingly, this added value for cooperation is quite significant up to very high values for $p$. Hence, whether cooperation is useful through interference management or increasing coverage, we observe that it leads to significant scalable degrees of freedom gains even in presence of harsh shadow fading conditions.
		\begin{figure}[H]	
		\centering
			\includegraphics[width = \columnwidth]{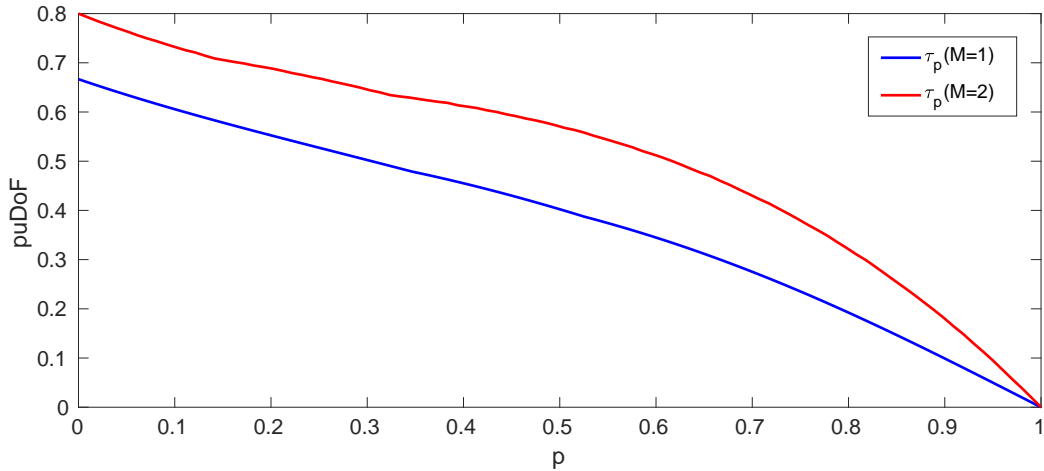}
			\caption{The plot shows $\tau_p{(M=1)}$ from~\eqref{eq:tauone} and the maximum value for the average puDoF achieved for $M=2$ by Algorithm $1$ as found by the simulation (Conjectured to be $\tau_p{(M=2)}$).}
\label{fig:CoMP}
\end{figure}
\section{Conclusion}\label{sec:conclusion}
We considered the problem of assigning messages to transmitters in a linear interference network with link erasure probability $p$, under a constraint that limits the number of transmitters $M$ at which each message can be available. For the case where $M=1$, we identified the optimal message assignment strategies at different values of $p$, and characterized the average per user DoF $\tau_p(M=1)$. For general values of $M\geq 1$, we proved that there is no message assignment strategy that is optimal for all values of $p$. We then introduced message assignment strategies for the case where $M=2$, and derived inner bounds on $\tau_p(M=2)$ that are asymptotically optimal as $p \rightarrow 0$ and as $p \rightarrow 1$. Finally, we presented an algorithm for $M=2$ that leads to the optimal average per user DoF under restriction to cooperative zero-forcing schemes. Further, we demonstrated the information-theoretic optimality of the presented algorithm for a wide class of network realizations. Simulation results were then used to affirm the intuition about the shifting role of cooperative transmission from interference management at low erasure probabilities to increasing coverage at high erasure probabilities. 

\bibliographystyle{IEEEtran}

\end{document}